\documentclass[12pt]{article}

\usepackage{amssymb,amsmath,amsfonts,amssymb}
\usepackage{graphics,graphicx,color}
\usepackage{empheq}

\def\@abssec#1{\vspace{.05in}\footnotesize \parindent .2in
{\bf #1. }\ignorespaces}

\def\endproof{{\ \vbox{\hrule\hbox{%
   \vrule height1.3ex\hskip0.8ex\vrule}\hrule
  }}\par}
\graphicspath{{/EPSF/}{Figures/}}
\DeclareGraphicsExtensions{.eps}

\newtheorem{theorem}{Theorem}[section]

\newtheorem{remark}[theorem]{Remark}
\newtheorem{hypothesis}[theorem]{Hypothesis}
\def \Rm {\mathbb R}
\def \Nm {\mathbb N}
\def \Cm {\mathbb C}
\def \Zm {\mathbb Z}
\def \Sm {\mathbb S}
\def \Tm {\mathbb T}
\newcommand{\eps}{\varepsilon}
\newcommand{\E}{\mathbb E}
\newcommand{\dsum}{\displaystyle\sum}
\newcommand{\dint}{\displaystyle\int}

\newcommand{\mH}{\mathcal H} 
\newcommand{\mK}{\mathcal K} \newcommand{\mL}{\mathcal L}
 \newcommand{\mN}{\mathcal N}
 \newcommand{\mT}{\mathcal T}
\newcommand{\mV}{\mathcal V}

\newcommand{\mn}{\mathfrak n}

\newcommand{\fH}{\mathfrak H}

\newcommand{\fa}{{\mathfrak a}}

\newcommand{\mh}{{\mathfrak h}}

\newcommand{\rH}{{\rm H}}

\newcommand{\tV}{{\mathtt V}}

\newcommand{\ind}{{\rm Index}}

\newcommand{\dt}{{\rm det}}

\newcommand{\cout}[1]{}

\newcommand{\sgn}[1]{\,{\rm sign}(#1)}

 \renewcommand{\arraystretch}{1.5}

\newcommand{\hplus}{h_{\tau}}
\newcommand{\hplusv}{h_{\tau v}}
\newcommand{\hminus}{h_{o}}
\newcommand{\hminusv}{h_{ov}}
\newcommand{\thplus}{{\tilde h_{\tau}}}
\newcommand{\thminus}{{\tilde h_{o}}}
\newcommand{\mhplus}{{\mh_{\tau}}}
\newcommand{\mhminus}{{\mh_{o}}}
\newcommand{\Mtop}{{M_{\tau}}}
\newcommand{\Ntop}{{N_\tau}}
\newcommand{\Mtriv}{{M_o}}
\newcommand{\Ntriv}{{N_o}}
\newcommand{\mplus}{m_\tau}
\newcommand{\mminus}{m_o}
\newcommand{\faplus}{{\fa_{\tau}}}
\newcommand{\faplusstar}{{\fa_\tau^*}}
\newcommand{\faminus}{{\fa_{o}}}
\newcommand{\faminusstar}{{\fa_o^*}}
\newcommand{\mnplus}{\mn_\tau}
\newcommand{\mnminus}{\mn_o}

\title{Topological protection of perturbed edge states}

\author{Guillaume Bal \thanks{Departments of Statistics and Mathematics and CCAM, University of Chicago, Chicago, IL 60637; guillaumebal@uchicago.edu}}

\begin{document}
 
\maketitle


\begin{abstract}
This paper proposes a quantitative description of the low energy edge states at the interface between two-dimensional topological insulators. They are modeled by continuous Hamiltonians as systems of Dirac equations that are amenable to a large class of random perturbations. We consider general as well as fermionic time reversal symmetric models. In the former case, Hamiltonians are classified using the index of a Fredholm operator. In the latter case, the classification involves a ${\mathbb Z}_2$ index. These indices dictate the number of topologically protected edge states.

A remarkable feature of topological insulators is the asymmetry (chirality) of the edge states, with more modes propagating, say, up than down. In some cases, backscattering off imperfections is prevented when no mode can carry signals backwards. This is a desirable feature from an engineering perspective, which raises the question of how back-scattering is protected topologically. A major motivation for the derivation of continuous models is to answer such a question.  

We quantify how backscattering is affected but not suppressed by the non-trivial topology by introducing a scattering problem along the edge and describing the effects of topology and randomness on the scattering matrix. Explicit macroscopic models are then obtained within the diffusion approximation of field propagation to show the following: the combination of topology and randomness results in un-hindered transport of the protected modes while all other modes (Anderson) localize.


%
%

\medskip



\end{abstract}
 

\renewcommand{\thefootnote}{\fnsymbol{footnote}}
\renewcommand{\thefootnote}{\arabic{footnote}}

\renewcommand{\arraystretch}{1.1}





\section{Introduction}


The characterization of phases of materials by  topological invariants rather than by symmetries and their spontaneous breaking constitutes a very active field of research in condensed-matter physics. These phases display fundamental properties that are immune to continuous changes in the material parameters unless the topological invariant ceases to be defined. Examples of such properties in two dimensional materials are the quantum Hall effect and the quantum spin Hall effect, which display unusual transport properties of electronic edge states at the interface between insulators \cite{Bernevig1757,chiu2016classification,PhysRevB.76.045302,PhysRevLett.98.106803,PhysRevLett.61.2015,RevModPhys.82.3045,PhysRevLett.95.146802,RevModPhys.83.1057}. Similar effects have also been predicted and observed in many photonic and mechanical structures \cite{fang2012realizing,fleury2016floquet,Fleury2014,hafezi2011robust,hafezi2013imaging,PhysRevLett.100.013904,
khanikaev2013photonic,lu2014topological,nash2015topological,
PhysRevA.78.033834,rechtsman2013photonic,umucalilar2011artificial,
wang2009observation,PhysRevLett.100.013905}.

The objective of this paper is to derive a continuous partial differential model that allows for a quantitative analysis of the properties of such edge (or interface) states and how topology influences their behavior under perturbations by random fluctuations. 

The existing models for edge states of topologically non-trivial materials  are typically obtained as follows \cite{chiu2016classification,Fruchart2013779,RevModPhys.82.3045,RevModPhys.83.1057}. Materials obeying translational invariance are characterized by a Hamiltonian $H(k)$ that depends continuously on a wavenumber $k$ living in a Brillouin zone $\Tm_2$. Such Hamiltonians are gapped by means of a mass term, an order parameter $M$. For any energy $E$ in the band gap, i.e., not in the spectrum of $H(k)$ for any $k\in \Tm_2$, the material is therefore an insulator. The mass term $M$, which may be a scalar quantity or a more general object, describes the (tunable) topology of the material: for certain values of $M$, say $M_R$, the material is topologically trivial, while for other values of that order parameter, say $M_L$, the material is non-trivial.  Edge states appear at the interface between materials characterized by $M_R$ and by $M_L$. Assuming a smooth transition from one value to another, the mass term $M(x)$ passes through a value  where the topological invariant is not defined, the gap is closed, the material ceases to be an insulator, and (metallic) edge states are allowed. The nature of the edge states and their topological protection can often be related to the topology of the insulators described by $M_{R,L}$. This is the bulk-boundary correspondence \cite{chiu2016classification,fukui2012bulk,Graf2013,RevModPhys.82.3045,PhysRevB.83.125109,RevModPhys.83.1057}.

The above results are based on the analysis of Hermitian bundles (parametrized by the order parameter $M$) of appropriate eigen-spaces of the Hamiltonian over the Brillouin zone $\Tm_2$ (or more general compact phase spaces in a `semi-classical' approach \cite[Section III.B]{chiu2016classification}). As such, they require that the material be invariant with respect to discrete translations. This renders the analysis of random fluctuations that naturally exist at all scales and therefore break the translational invariance quite difficult. 

To obtain quantitative models for the influence of non-periodic random fluctuations on edge states, we have to leave the realm of continua described by a Brillouin zone.  More precisely, we need to leave the commutative setting of Fourier (or phase-space) multipliers on the Brillouin zone and the topology of Hermitian bundles, and consider instead the non-commutative setting of operators acting on the physical variables, where the topology is given by the indices of appropriate Fredholm operators. Such a successful framework was introduced in \cite{bellissard1986k,bellissard1994noncommutative} to model the quantum hall effect. It was extended to more general topological material, including the analysis of edge and interface states for discrete Hamiltonians, their perturbations by randomness, and the bulk-boundary correspondence in a large body of works in the mathematical literature; see \cite{Avila2013,bourne2016chern,doi:10.1063/1.4902377,1751-8121-44-11-113001,prodan2016bulk} and their numerous references.
This index approach is related to the notion of relative index of projections developed in \cite{PhysRevLett.65.2185,avron1994,AVRON1994220} to analyze the quantum Hall effect.

\medskip

The main objective of this work is to introduce continuous partial differential models, also of ``index'' type, to quantify the transport properties of edge states and assess how such properties are affected by topological constraints and random perturbations. In particular, we describe how the backscattering of edge modes is affected by topology but not always suppressed in the presence of randomness. 
For concreteness, we focus on an electronic application although the mathematical model, a system of Dirac equations, applies to the analysis of transport in topological photonics; see above references.

\medskip

{\bf Partial differential model.}
The framework we propose here is best motivated by a concrete example related to graphene. We refer the reader to \cite{RevModPhys.81.109,Fruchart2013779,RevModPhys.82.3045,prodan2016bulk,RevModPhys.83.1057} for the details and context. Unperturbed graphene is modeled by a Hamiltonian $H(k)$, $k=(k_x,k_y)$ in a Brillouin zone $\Tm_2$, a two-dimensional torus, describing bulk states at different sub-lattices ($A$ or $B$) and spin configurations (up and down). Two wavenumbers $\xi K$ for $\xi=\pm1$, the Dirac points, are special  in that  the conduction and valence bands meet exactly at these two points for an energy $E$ normalized to $0$. Valleys are then described by wavenumbers $k=\xi K+q$ in the vicinity of these two Dirac points. Focusing on one valley, say $\xi=1$, for a given spin, say up, and linearizing the Hamiltonian in the vicinity of the Dirac point gives a Dirac operator for a Fermi velocity $v$ written in the Fourier domain as 
\begin{displaymath}
  v (q_x\sigma_1 + q_y\sigma_2)= v \left(\begin{matrix}
   0&q_x-iq_y\\q_x+iq_y&0
\end{matrix}\right)
\end{displaymath}
and describing quantum states at the sub-lattices $A$ and $B$ (first and second components, respectively, of the spinor in $\Cm^2$ to which the Hamiltonian applies), where $(\sigma_1,\sigma_2,\sigma_3)$ are the standard Pauli matrices. The above linear dispersion relation, which is very common in the physical literature \cite{RevModPhys.81.109,Fruchart2013779,RevModPhys.82.3045,RevModPhys.83.1057}, is valid for low frequencies $v q$, or equivalently low energies. This is the regime considered in this paper. 

The above model is not gapped and graphene is metallic at energy $E=0$. One mechanism to gap graphene or graphene-like structures corresponds to an asymmetry between the sub-lattices $A$ and $B$, which results in the Hamiltonian for the valley $\xi=1$,
\begin{displaymath}
  H_1(q) = v (q_x\sigma_1+q_y\sigma_2 + m_1 \sigma_3).\end{displaymath}

Let us assume that the valley $\xi=-1$ is gapped as well but with a mechanism that results in a mass $m_2$ not necessarily equal to $m_1$ and a corresponding Hamiltonian $H_2(q)$. It may be shown \cite{Fruchart2013779,prodan2016bulk} quite generally, including for the celebrated Haldane model \cite{PhysRevLett.61.2015}, that the (first) Chern number $c_1=\frac12(\sgn{m_2}-\sgn{m_1})$ is a topological invariant for the Hamiltonian $H(k)$ defined  for $k$ over the whole Brillouin zone $\Tm_2$ and whose linearization in the vicinity of $\pm K$ is given by $H_1(q)\oplus H_2(q)$. The Chern number is physically relevant as it can be related by the Kubo formula to the quantum Hall conductivity, which can be observed experimentally \cite{RevModPhys.81.109,Fruchart2013779,RevModPhys.82.3045,prodan2016bulk,RevModPhys.83.1057}. Note that such a model breaks time reversal symmetry. 

Here and below, we define the direct sum of operators $h_j:\rH_j\to \tilde \rH_j$, $j=1,2$, as the operator $h_1\oplus h_2$ from $\rH_1\oplus \rH_2$ to $\tilde \rH_1\oplus\tilde \rH_2$ defined by $(h_1\oplus h_2) (u,v)=(h_1u,h_2v)$, i.e., formally the operator ${\rm Diag}(h_1,h_2)$.

A standard route to the derivation of edge states now assumes that $m_j=m_j(x)$ with $x$ a macroscopic spatial variable. Then $x>0$ and $x<0$ correspond to two materials in (possibly) different topological phases. Assume $m_j(x)$  continuous with, say,  $m_1(x)\to\pm m_\alpha>0$ changing signs as $x\to\pm\infty$ while $m_2(x)\to m_\beta>0$ does not change signs asymptotically in the same limits. Then, $c_1=c_1(x)=\frac12(\sgn{m_2(x)}-\sgn{m_1(x)})$ changes values as $x$ runs from $-\infty$ to $+\infty$ and hence must jump somewhere, say at $x=0$, where the material is metallic. An edge state localized in the vicinity of $x=0$ may then appear and propagate in the transverse direction $y$. This is confirmed by the classical analysis we will come back to in detail in section \ref{sec:scat} of a Hamiltonian of the form
\begin{equation}\label{eq:example}
  H(q_y) = v \big(\frac1i\partial_x\sigma_1+q_y\sigma_2+m_1(x)\sigma_3\big) \oplus v \big(\frac1i\partial_x\sigma_1+q_y\sigma_2+m_2(x)\sigma_3\big).
\end{equation} 
This Hamiltonian may be represented as a $4\times4$ matrix of operators and applies to spinors of the form $(\psi_{A,\xi=1},\psi_{B,\xi=1},\psi_{A,\xi=-1},\psi_{B,\xi=-1})^t$.
The bulk-edge correspondence provides a link between the topological bulk properties $c_1(\pm\infty)$ of the two materials at $x>0$ and $x<0$ and the number of edge modes $c_1(+\infty)-c_1(-\infty)$ concentrated in the vicinity of where the topological number $c_1(x)$ jumps.

Once written in the physical domain (where the multiplier $q_y$ is replaced by $\frac1i\partial_y$), we obtain an unperturbed partial differential model for the interface, a system of Dirac equations, which no longer requires the assumption of translational invariance.


\medskip

{\bf Edge models.} Our aim is to analyze the edge states in the vicinity of $x=0$. Since both domains $\pm x>0$  are insulators, we now introduce operators that neglect (prevent) bulk propagation and the resulting continuous spectra for energies not lying in the band gap. This greatly simplifies our analysis. The work \cite{B-BulkInterface-2018} shows that such simplifications can be avoided at the price of significant functional technicalities we do not consider here.

 Let us focus on  the mass term $m(x)$ that changes signs. We wish to obtain a function that faithfully describes the change of topology (i.e., changes sign in the vicinity of $0$) while confining the system to the vicinity of $x=0$. The simplest example is $\mplus(x)=\lambda x$ for some $\lambda>0$. A large $\lambda$ corresponds to a sharp transition from one topology to the other. The unperturbed edge states  are now modeled by Hamiltonians of the form (normalizing the Fermi velocity $v=1$ for the rest of the paper):
\begin{displaymath}
   \thplus = \frac1i \partial_x \sigma_1 + \frac1i \partial_y \sigma_2 + \mplus(x) \sigma_3.
\end{displaymath}
The mass term that does not change signs may be treated similarly, with the main difference that $m=\mminus(x)$ has a fixed sign at $\pm\infty$. Our choice of a model for $\mminus(x)$ will be a smooth function of $x$ with a behavior at $\pm\infty$ of the form $\lambda|x|$ so that $\sgn{\mminus(x)}$ is constant.  The Hamiltonian for the Haldane model corresponding to a transition of $c_1(x)$ from a trivial phase $c_1(x)=0$ to a non-trivial phase $c_1(x)=1$ with a Hamiltonian in physical variables given by $\thplus\oplus \thminus$, where $\thminus$ is the Hamiltonian associated to the mass term $\mminus(x)$.

    The physical intuition for such a model is clear: we assume that the energy range where we operate the material is very small compared to the bulk gap generated by the mass terms  $m_{\tau,o}(\pm\infty)$ and so formally send the latter to infinity to avoid mathematical difficulties that are irrelevant to characterize the propagation of the interface modes.

More general materials (than the above Haldane model) may be represented as more general direct sums of block Hamiltonians of the form given above corresponding to the different species present in the system (sub-lattice as described above, valleys, spin, or any other internal degree of freedom; see \cite[III.C.1]{chiu2016classification}). Once we have obtained such a family of unperturbed Hamiltonians $\tilde H_0$, we can perturb them by a large class of random fluctuations $V$, with $V$ a Hermitian operator.

Let us reiterate that the models considered here are low energy approximations for wavenumbers close to Dirac points. Non-trivial bulk materials are characterized by non-trivial integers, the Chern numbers, which may be represented as integrals of curvature forms over the Brillouin zone. Such integers depend on the singular behavior of the curvature form in the vicinity of critical points, the Dirac points ($\pm K$ in the preceding model). Hamiltonians generically take the form of Dirac operators in the vicinity of the Dirac points. These Hamiltonians are our starting point. They encode the Chern numbers, and hence the topology of the materials typically considered in the physical literature \cite{Fruchart2013779}, by means of mass terms (of the form $m_{\tau,o}$ in the preceding model). We assume a continuous transition of the mass terms from one material to the other. We also assume that masses tend to infinity away from the interface in order to focus on edge properties, thereby eliminating the possibility of escape into the bulk and its analysis; see \cite{B-BulkInterface-2018}. This provides the starting point of our analysis, namely an operator of the form $\tilde H_0$ written as the direct sum of $2\times2$ Dirac operators of the form of $\thplus$ above.

\medskip

{\bf Topological classification.} Hamiltonians in the non-commutative setting are typically mapped to Fredholm operators by spectral calculus, whose index reflects the non-trivial topology of the problem; see \cite{avron1994,bourne2016chern,prodan2016bulk} and references there. The Dirac operators considered here can be mapped to similar Fredholm operators as shown in \cite{B-BulkInterface-2018}. We will instead leverage the asymptotic behavior of $m_{\tau,o}$ at $\pm\infty$ and propose a somewhat simpler classification.

To classify the above Hamiltonian $\thplus$ (or $\tilde H_0$), we introduce in section \ref{sec:TI} a family of Fredholm operators given by the regularization $\tilde D_v=\sigma_1\otimes\tilde H_0- vy\sigma_2\otimes I$ for some $v>0$ arbitrarily small, which should be thought of as $v=0^+$. This regularization is necessary to `compactify' the infinite domain $\Rm^2$ and bears some similarities with the topological method based on Green's functions presented in \cite{volovik2009universe}; see also \cite{fukui2012bulk}. The operator $\tilde D_v$ is a Dirac operator to which a Fredholm index can be assigned. 

The class of edge Hamiltonians considered in this paper is introduced in section \ref{sec:edgeH}. In section \ref{sec:TI}, we show that the index of these Hamiltonians is: (i) not modified in the presence of a large class of (random) perturbations; (ii) equal to $\Mtop-\Ntop$, with $\Mtop$ and $\Ntop$ the numbers of (protected) edge states propagating in the $\pm y$ direction in the vicinity of $x=0$, respectively; and (iii) two Hamiltonians with the same index are connected by a continuous path of appropriate Fredholm operators. 

Thus far, the Hamiltonians, such as $\thplus$ above, do not respect time reversal symmetry (TRS). In the presence of TRS of fermionic type introduced in section \ref{sec:mod2}, we verify that $\Mtop=\Ntop$ and so the above index vanishes. The remarkable result obtained in \cite{PhysRevB.76.045302,PhysRevLett.98.106803,RevModPhys.82.3045,PhysRevLett.95.146802} shows that because of the fermionic TRS, edge modes are still topologically (or algebraically) protected when they come in an odd number of pairs. Introducing the $\Zm_2$ index given by $\Mtop$ mod $2$, we obtain again that such an index is immune to a large class of random perturbations while two operators in the same class are connected by a path of Fredholm operators preserving the TRS.

\medskip

{\bf Scattering theory and diffusion approximation.} The edge Hamiltonians considered here involve minor modifications (primarily the behavior of the mass terms at infinity to concentrate the analysis to the vicinity of the interface and the resulting classification as Fredholm operators of regularized Dirac type) of low energy models readily available in the literature \cite{Fruchart2013779,prodan2016bulk,RevModPhys.83.1057}. The main motivation for their introduction is to {\em quantify} the interaction of random fluctuations and topology on the edge modes, and in particular obtain quantitative descriptions of the physically relevant notions of transmission and reflection (backscattering).

 Operators such as $\thplus$ are shown to be decomposed (for each wavenumber $\zeta$, the dual variable to $y$) into an infinite number of edge modes providing an appropriate basis of functions in $L^2(\Rm_x;\Cm^{\cal N})$, where ${\cal N}$ is the dimension of the spinors to which the Hamiltonian is applied. Some of these modes are the edge modes among the $\Mtop$ and $\Ntop$ that characterize the topology of $\tilde D_v$. Combining them with the other modes provides a basis to describe a scattering theory as follows.

For a given energy level $E$, a finite number of these edge modes are propagating while the rest are evanescent. In the presence of random fluctuations $\tilde V$ coupling the propagating modes (see Hypothesis \ref{hyp:prop} in section \ref{sec:scat}), the amplitudes of said modes satisfy a closed system of equations (in the $y$ variable). Such a scattering theory is introduced in section \ref{sec:scat}. Edge transport is then characterized by a scattering matrix composed of reflection and transmission coefficients. Conductance in such systems is then physically proportional to the trace of the transmission matrix. In the topologically trivial setting, Anderson localization shows that such a conductance decays exponentially as the thickness of the slab of random perturbations increases. We will show that the conductance is at least equal to the index of the Dirac operator $\tilde D_v$ in general and at least equal to the mod 2 index in the presence of TRS. This validates the intuition that non-trivial topology prevents (complete) Anderson localization of the edge states. 

Often associated to the absence of localization is the absence of back-scattering. The latter does not hold in general. For sufficiently low energy levels, the absence of back-scattering is certainly observed in some cases.  However, for sufficiently high energy levels, the number of edge modes is given by the protected modes plus a number of pairs of modes that is energy dependent. Scattering among these modes is triggered by the random fluctuations $\tilde V$.  In the diffusive regime, the scattering coefficients satisfy (in some cases) explicit quantitative diffusion (Brownian-type) motion as a function of thickness $L$ of the random medium. In this configuration, transmission is guaranteed by topology and all modes experience back-scattering except specific, medium-dependent, modes that are indeed reflection (back-scattering) free. In the presence of large randomness, we obtain the striking feature that $\Mtop-\Ntop$ ($\Mtop{\rm mod}2$ in the TRS setting) modes are allowed to transmit while every other mode (Anderson) localizes. The details of the derivation are presented in sections \ref{sec:scatTR} and  \ref{sec:diff}.



\medskip

The Hamiltonians considered here are low-frequency (energies close to the Fermi energy) approximations of tight-binding Hamiltonians derived for translation invariant materials. For recent derivations of edge states and their topological protection from explicit potentials in a Schr\"odinger equation with appropriate non-periodic perturbations, we refer the reader to \cite{Fefferman17062014,Fefferman2016}.  The stability properties of edge states we obtain in this paper in the TRS setting are consistent with those derived in \cite{Sadel2010} for randomly perturbed one dimensional Dirac (continuous) models. 

\section{Edge Hamiltonians}
\label{sec:edgeH}

The introduction section recalled the Dirac  equation \eqref{eq:example}  to model low energy edge modes. We now consider more general  models of  topologically trivial and non-trivial edge modes propagating in the $\pm y$ directions. We represent the Dirac operators in an equivalent, more convenient, basis. General edge Hamiltonians are then written as a direct sum of elementary blocks, which are not constrained to satisfy any translation invariance and can thus be perturbed by a large class of random fluctuations. 

\medskip

{\bf Topologically non-trivial block.} The main elementary ($2\times2$ block) operator in two space dimensions carrying an edge mode is given by:
\begin{equation}\label{eq:h+}
  \hplus = \left(\begin{matrix}
  \frac 1i \partial_y & \faplusstar \\ \faplus & -\frac 1i\partial_y
\end{matrix}\right) = \frac1i \partial_y \sigma_3+\frac1i \partial_x \sigma_2 + \mplus(x)\sigma_1.
\end{equation}
Here, $\faplus=\partial_x+\mplus(x)$ and its formal adjoint $\faplusstar=-\partial_x+\mplus(x)$. We assume that $C^\infty(\Rm)\ni \mplus(x)\to \pm\infty$ as $x\to\pm\infty$ and that $\frac{\mplus'(x)}{\mplus(x)}\to0$ as $|x|\to\infty$. To simplify the functional setting, we assume that $0<\lambda_0^{-1}<\frac{|\mplus(x)|}{|x|}<\lambda_0$ for $|x|>1$. A typical example is $\mplus(x)=\lambda x$ for $\lambda>0$ so that $\faplus$ is a rescaled version of the standard creation operator and $\faplus\faplusstar$ is related to the harmonic oscillator of quantum mechanics.

 We verify that $\faplus$ is bounded from $\fH_1(\Rm;\Cm^2)$ to $\fH_0(\Rm;\Cm^2)=L^2(\Rm;\Cm^2)$ with $\fH_1(\Rm;\Cm^n)$ the Hilbert space defined by the norm $(\|u\|^2_{L^2(\Rm)} + \|xu\|^2_{L^2(\Rm)}+\|\partial_x u\|^2_{L^2(\Rm)})^{\frac12} <\infty$ for each of the $n$ components.   Consequently, its adjoint operator $\faplusstar$ is bounded from $\fH_0$ to $\fH_1^*$, the dual space to $\fH_1$; here and below we use $\fH_1$ for $\fH_1(\Rm;\Cm^n)$ when $n$ is obvious from the context. The formal adjoint operator $\faplusstar$ is also bounded from $\fH_1$ to $\fH_0$ so that $\faplus$ is bounded from $\fH_0$ to $\fH_1^*$ as well. To simplify notation, we use the same expression $\faplus$ and $\faplusstar$ for these operators defined on different domains, so that both $\faplus\faplusstar$ and $\faplusstar\faplus$ are bounded from $\fH_1^*$ to $\fH_1$. 

The above Hamiltonian $\hplus$ is the same as $\thplus$ described  in the introduction but written in a different basis. Let $Q_2=\frac1{\sqrt2}\left({\scriptsize\begin{matrix}
 1&1\\i&-1
\end{matrix}}\right)$ be the unitary matrix whose columns are the eigenvectors of $\sigma_2$. We verify that $Q_2^*\thplus Q_2=\hplus$. This new basis is more convenient as $\partial_x$ and $m_{\tau,o}(x)$ appear in the same matrix entries; see also \cite{Fruchart2013779}. 

We shall see in section \ref{sec:scat} that the above operator carries one edge mode propagating in the positive $y$ direction without dispersion as well as an infinite number of pairs of dispersive modes.  Standard calculations \cite{Fruchart2013779,RevModPhys.83.1057} show that when $\mplus(x)=\lambda x$,  a solution to $\hplus\psi=E\psi$ for $E\in\Rm$ is given by 
\begin{equation}
\label{eq:edgemode}
  \psi(x,y) =  c e^{iEy} e^{-\frac\lambda2 x^2} \left(\begin{matrix}
  1\\0
\end{matrix}\right),
\end{equation}
with $c$ a normalizing constant. This is the typical example of an edge mode concentrated in the vicinity of $x=0$ with a linear dispersion relation $E(\zeta)=\zeta$ for $\zeta$ the dual (Fourier) variable to $y$ with a group velocity $\frac{\partial E}{\partial \zeta}=1$ (the Fermi velocity $v$ being normalized to $1$).

Here and below, we use the notation ${}^*$ to represent Hermitian conjugation and $\bar{\ }$ to represent complex conjugation (as opposed to the notation ${}^\dagger$ and ${}^*$, respectively, that is more standard in the physics literature).

We also consider edge modes propagating in the opposite direction (toward negative values of $y$). They are supported by the operator
\begin{displaymath}
  - \hplus = \left(\begin{matrix}
  -\frac 1i \partial_y & -\faplusstar \\ -\faplus & \frac 1i\partial_y
\end{matrix}\right) ,
\end{displaymath}
where the direction of time changed in a time-dependent Schr\"odinger equation.  Such modes are also modeled by the operator $\bar \hplus$, which may be seen as the time reversal conjugate to $\hplus$; see section \ref{sec:mod2}. These two operators, $\bar \hplus$ and $-\hplus$, are the same operator written in different bases since $\sigma_3\bar \hplus\sigma_3=-\hplus$ (and $\sigma_3^{-1}=\sigma_3$). Since some calculations to follow are simpler with $-\hplus$, we use this choice of a representation.


\medskip

{\bf Topologically trivial blocks.} 
We finally consider the case of non topologically protected edge modes. Such modes are modeled by a localizing mass term $\mminus(x)$ that does not change sign at infinity. We assume that $C^\infty(\Rm)\ni \mminus(x)\to +\infty$ as $x\to\pm\infty$. We also assume that $0<\lambda_0^{-1}<\frac{|\mminus(x)|}{|x|}<\lambda_0$ for $|x|>1$; typically $\mminus(x)$ is a smooth version of (or exactly that if we relax the regularity constraint) $\lambda |x|$ for $\lambda>0$.  We could similarly have a term $\mminus(x)$ converging to $-\infty$ at $\pm\infty$. What is topologically relevant is that $\mminus$ has the same sign at $\pm\infty$ while $\mplus$ changes sign from $-\infty$ to 
$+\infty$.

We then define the building block
\begin{equation}\label{eq:h-}
  \hminus = \left(\begin{matrix}
  \frac 1i \partial_y & \faminusstar \\ \faminus & -\frac 1i\partial_y
\end{matrix}\right), \qquad \faminus = \partial_x+\mminus(x).
\end{equation}
The time reversal conjugate of such a block is then $-\hminus=\sigma_3\bar \hminus\sigma_3$. The operators $\faminus$ and $\faminusstar$ share the same functional setting as $\faplus$ and $\faplusstar$. Note that no mode of the form \eqref{eq:edgemode} with linear dispersion relation exists for the above Hamiltonian.

\medskip

{\bf Unperturbed edge Hamiltonian.} The unperturbed Hamiltonian describing the propagation of the edge modes is then given more generally by the direct sum of copies of the preceding building blocks:
\begin{equation}\label{eq:H0}
   H_0 = \hplus^{\oplus \Mtop} \oplus (-\hplus)^{\oplus \Ntop} \oplus \hminus^{\oplus \Mtriv} \oplus (-\hminus)^{\oplus \Ntriv}.
\end{equation}
Such Hamiltonians now act on spinors of size $\mN=2(\Mtop+\Ntop+\Mtriv+\Ntriv)$. This generalizes the operator in \eqref{eq:example} (written in the Fourier domain for the variable $y$) where we had $\Mtop=\Mtriv=1$. Here, $\Mtop$ models the Chern number generalizing $c_1(+\infty)$ introduced in the introduction while $\Ntop$ models the the Chern number generalizing $c_1(-\infty)$. The bulk-edge correspondence then states that the topologically relevant number of protected edge modes should be $\Mtop-\Ntop$. This will be verified in the next sections.

There is no fundamental reason to assume that the mass terms $m_{\tau,o}(x)$ are the same for all the propagating modes. To (slightly) simplify notation, however, we assume that they are indeed the same and that $\mplus^2(x)=\mminus^2(x)$ for $|x|\geq1$, say. This is the basic setting for the modeling of edge modes, with a collection of $\Mtop$ upward propagating modes with (approximately) linear dispersion, $\Ntop$ downward propagating modes also with (approximately) linear dispersion, and as we shall see, a large class of other dispersive modes propagating upwards and downwards (i.e., with positive and negative currents) in equal numbers.

\medskip

{\bf General edge Hamiltonian.}
With this definition, the class of Hamiltonians we consider in this paper are of the form
\begin{equation}\label{eq:edgeH}
  H = H_0 + V,
\end{equation}
where $V$ is a perturbation that models a wide class of (random) fluctuations. We will be more precise on our assumptions on $V$ in the next sections.

\medskip

The objective of this paper is to analyze and quantify the robustness of the edge modes with respect to the perturbation $V$. We distinguish two types of systems. One is the general setting of time reversal breaking Hamiltonians as described above. We will show that an index given by $\Mtop-\Ntop$ characterizes the classes of homotopically equivalent Hamiltonians for a given dimension of spinors ${\cal N}= 2(\Mtop+\Ntop+\Mtriv+\Ntriv)$. 

The second regime is that of Hamiltonians satisfying a fermionic Time Reversal symmetry. In such a setting,  $\Mtop=\Ntop$ and the index vanishes. We will show that the Hamiltonians with a given dimension ${\cal N}$ are separated into two homotopy classes given by the $\Zm_2$ index  $\Mtop{\rm mod}\ 2$. 


%
\section{Topological invariants}
\label{sec:TI}

Hamiltonians in the non-commutative setting such as those presented in the preceding section (operators $H$ corresponding to different perturbations $V$ or different mass terms $m(x)$ no longer commute) need to be assigned topological invariants. A well-established procedure to do so is based on the notion of Fredholm modules and on Fredholm operators constructed from the Hamiltonian by spectral calculus; see \cite{bourne2016chern,prodan2016bulk} for definitions and details, as well as \cite{B-BulkInterface-2018} for continuous Dirac Hamiltonians. Here, we use the specific structure of the mass terms $m(x)$ to introduce a simpler topological classification. We first focus on  the setting in which the Hamiltonian is not necessarily time reversal symmetric and then the setting where the  fermionic TRS holds.

\subsection{Index theory for general edge Hamiltonians}

The edge Hamiltonians described in \eqref{eq:edgeH} are now mapped to Fredholm operators to which an index may be assigned. 
Let us first focus on the block $\hplus$. Such a Hamiltonian may be written in the Fourier domain $y\to\zeta$, where it has the following expression
\begin{displaymath}
  \hat \hplus(\zeta) = \left(\begin{matrix}
   \zeta & \faplusstar \\ \faplus & -\zeta
\end{matrix}\right).
\end{displaymath}
For each $\zeta$, this is a one-dimensional Hamiltonian with purely discrete spectrum  and a Fredholm operator from $\fH_1$ to $\fH_0$ thanks to the confinement provided by $\mplus(x)$. Since $\zeta$ is continuous in $\Rm$, this shows that $\hplus$ written in an appropriate functional setting on $\Rm^2$ will have essential (continuous) spectrum in the vicinity of $0$ (and hence cannot be Fredholm). In order to classify Hamiltonians, we use a classical regularization that renders all these Hamiltonians Fredholm operators by 'compactifying' the variable $y$. The simplest regularization consists of adding a term of the form $vy$ appropriately for $v>0$ (arbitrarily small).  

We should think of $v=0^+$ so that the physical description of the Hamiltonian is preserved in the limit $v\to 0^+$. Only the sign of $v$ matters to have a well defined index.  Specifically, we introduce the regularization
\begin{displaymath}
  D_v = \sigma_1\otimes H - \sigma_2\otimes vy = \left(\begin{matrix}
  0&H_v^*\\H_v&0
\end{matrix}\right),\quad H_v=H+\frac1i vy.
\end{displaymath}
Here, $H_v$ is a non-Hermitian operator because of the presence of the $\frac1ivy$ term, and $D_v$ is the associated (Hermitian) Dirac operator. Our edge Hamiltonians are then classified according to the index of the above Dirac operator defined as
\begin{equation}\label{eq:index}
  \ind \ D_v := \ind\ H_v = \Mtop-\Ntop.
\end{equation}
We now justify the use of such an index.

Let $\hplusv$ be the regularization of $\hplus$:
\begin{displaymath}
  \hplusv = \left(\begin{matrix}
   \frac1i \fa_v & \faplusstar \\ \faplus & \frac1i \fa_v^*
\end{matrix}\right),\qquad \fa_v = \partial_y+vy.
\end{displaymath}
We recognize in $\hplusv$ a similar operator to the one defined in \cite[Prop. 19.2.9]{H-III-SP-94} with $n=2$. 

It is not difficult to verify that it is a Fredholm operator with index equal to $1$ from  $\mH_1(\Rm^2;\Cm^2)$ to $\mH_0(\Rm^2;\Cm^2)$, where $\mH_0=L^2((\Rm^2;\Cm^2))$ while $\mH_1(\Rm^2;\Cm^n)$ is the Hilbert space with norm $(\|(1+|x|+|y|)u\|^2+\|\partial_x u\|^2+\|\partial_y u\|^2)^{\frac12}$ for each of the $n$ components. These are the two-dimensional counterparts of $\fH_0$ and $\fH_1$, respectively. Here, $\|\cdot\|$ is the usual $L^2(\Rm^2)$ norm. We use the notation $\mH_j$ for $\mH_j(\Rm^2;\Cm^n)$ when $n$ is clear from the context. As in the one dimensional case, both $\hplusv$ and $\hplusv^*$ are Fredholm operators (with indices $+1$ and $-1$, respectively) from $\mH_1$ to $\mH_0$ as well as from $\mH_0$ to $\mH_1^*$, the dual to $\mH_1$.

If $\hminusv$ is defined similarly with $\faplus$ replaced by $\faminus$, then we obtain a Fredholm operator (in the same topologies) with a vanishing index.  Similarly, $(-\hminus)_v=-\hminus+\frac1ivy$ is also a Fredholm operator with vanishing index.  Finally, we verify that 
\begin{displaymath}
  (-\hplus)_v= \left(\begin{matrix}
  \frac 1i (vy-\partial_y) & -\faplusstar \\ -\faplus & \frac 1i(vy+\partial_y)
\end{matrix}\right) = \left(\begin{matrix}
 \frac1i \fa_v^* & -\faplusstar \\ -\faplus & \frac1i \fa_v
\end{matrix}\right) = -\hplusv^*,
\end{displaymath}
is a Fredholm operator with an index equal to $-1$, as is $(\bar \hplus)_v=\sigma_3(-\hplus)_v\sigma_3$. Note that $(-\hplus)_v$ can be continuously deformed to $(\bar \hplus)_v$ along a path of (non Hermitian) Fredholm operators by defining $\sigma(t)={\rm Diag} (1,e^{i\pi t})$, which continuously deforms the identity matrix to $\sigma_3$ in the space of unitaries. The path is given by $\sigma^{-1}(t)(-\hplus)_v\sigma(t)$. As indicated above, we could therefore have used $\bar \hplus$ as the building block for negatively propagating edge modes in \eqref{eq:H0} instead of $(-\hplus)$.

This proves the validity of the index \eqref{eq:index} when $H$ is replaced by $H_0$, a Fredholm operator of index $\Mtop-\Ntop$ from $\mH_1(\Rm^2;\Cm^{{\cal N}})$ to $\mH_0(\Rm^2;\Cm^{{\cal N}})$. Now, for $V$ any relatively compact perturbation of $H_{0v}$, i.e., such that $(\lambda-H_{0v})^{-1}V$ is compact as an operator defined in ${\cal L}(\mH_0(\Rm^2;\Cm^{{\cal N}}))$ for one $\lambda\in\Cm$ (and hence all $\lambda$ in the resolvent set of $H_{0v}$), we obtain the standard result that the index is invariant  \cite[Chapter 19]{H-III-SP-94}: $\ind H_v=\ind H_{0v}$. This justifies the definition of the index \eqref{eq:index} for the operators $H_{v}$ considered in this paper. We observe that any local multiplication by $V(x,y)$ for $V$ bounded on $\Rm^2$ provides such a relatively compact perturbation. Note that the index is also independent of the choice of $v>0$.

It is a classical (deep) result that Fredholm operators on Hilbert spaces with the same index (possibly defined on different (separable) Hilbert spaces since all such spaces can be identified) are path-connected; see \cite[Theorem 3.40]{bleecker2013index} and the comments before that (Atiyah-J\"anich) theorem. Since we will need a similar result in the TRS setting where a general theory (which undoubtedly applies) may be harder to find, we propose below a simple standard construction of the path. The proof also displays explicitly why a material with $\Mtop=\Ntop=1$ that may look non-trivial topologically is in fact trivial (in that sense) and equivalent to a material with $\Mtriv=\Ntriv=1$.

\begin{theorem}
\label{thm:topo} Let $H_1$ and $H_2$ be two Hamiltonians of the form \eqref{eq:edgeH} with the same spinor dimension $\mN$.  Let $H_{1v}$ and $H_{2v}$ be their Fredholm regularizations and $V$ be a relatively compact perturbation as described above.  Assume that $\ind H_{1v}=\ind H_{2v}$. 

Then there exists a continuous family of Fredholm operator $F_t$ for $0\leq t\leq 1$ such that $F_0=H_{1v}$ and $F_1=H_{2v}$.
\end{theorem}
\begin{proof}
We wish to show that $\Mtop-\Ntop$ indeed characterizes the above class of edge Hamiltonians. 
Since $V$ is not affecting the index, $H_{0v}+tV$ for $t\in[0,1]$ provides a continuous family of Fredholm operators linking $H_{0v}$ to $H_v$.
We consider  $D_{0v}$ the regularization of the unperturbed operator $H_0$ defined in \eqref{eq:H0}. 
Let $\Mtop>0$ and $\Ntop>0$ and consider a sub-block of $H_0$ defined by the pair $f=\hplus\oplus (- \hplus)$ with $f_{v}$ its Fredholm regularization. Let now $g=\hminus\oplus (- \hminus)$ be a similarly defined block of trivial operators and $g_v$ its Fredholm regularization. More precisely, let us define
\begin{displaymath}
  f_v = \left(\begin{matrix}
 \hplusv & 0 \\ 0 & -\hplusv^*
\end{matrix}\right),\quad D_{f} = \left(\begin{matrix}
  0&f_v^* \\ f_v & 0
\end{matrix}\right),
 \qquad g_v = \left(\begin{matrix}
 \hminusv & 0 \\ 0 & -\hminusv^*
\end{matrix}\right),\quad D_{g} = \left(\begin{matrix}
  0&g_v^* \\ g_v & 0
\end{matrix}\right).
\end{displaymath}
Here, $f$ corresponds to $\Mtop=1$ and $\Ntop=1$ while $g$ corresponds to $\Mtriv=1$ and $\Ntriv=1$. We want to show that the two Fredholm operators $f_v$ and $g_v$ are linked by a continuous path of Fredholm operators and equivalently (if one insists on working with Hermitian operators) that the Hermitian Dirac operators $D_f$ and $D_g$ are linked by a continuous path of Hermitian Fredholm operators. This is done as follows. Let us define the rotation (by $-t$)
\begin{displaymath}
  R_t = \left(\begin{matrix}
  c_t & s_t \\ -s_t & c_t 
\end{matrix}\right), \qquad c_t=\cos\ t,\quad s_t=\sin\ t.
\end{displaymath}
Consider the family of operators
\begin{displaymath}
  F_t =\left(\begin{matrix}
  h_v & 0 \\ 0 & 1 
\end{matrix}\right)  \left(\begin{matrix}
  c_t & s_t \\ -s_t & c_t
\end{matrix}\right) \left(\begin{matrix}
  1 & 0 \\ 0 & -h_v^*
\end{matrix}\right) = \left(\begin{matrix}
  c_t h_v & -s_t h_v h_v^* \\ -s_t & -c_t h_v^*
\end{matrix}\right),
\end{displaymath}
for $h_v$ being  $\hplusv$ for $F_0=f_v$ or $\hminusv$ for $F_0=g_v$. We observe that this is a family of Fredholm operators (as a composition of Fredholm operators) from $\mH_0\oplus\mH_1$ to $\mH_1^*\oplus\mH_0$; see, e.g., \cite[Corollary 19.1.7]{H-III-SP-94}.
The family continuously links $F_0$ to
\begin{displaymath}
  F_{\frac\pi2} = -\left(\begin{matrix}
  0&h_v h_v^* \\ 1&0
\end{matrix}\right), \qquad h_{\alpha v} h_{\alpha v}^* =\left(\begin{matrix}
    \fa_v\fa_v^* + \fa_\alpha^* \fa_\alpha & 0 \\ 0 & \fa_v\fa_v^*+\fa_\alpha\fa_\alpha^*
\end{matrix}\right),\quad \alpha=\tau,o.
\end{displaymath}

Now, we observe that $\hplusv\hplusv^*$ and $\hminusv\hminusv^*$ are homotopic Fredholm operators. Indeed
\begin{displaymath}
  \faplus\faplusstar = -\partial_x^2 + \mplus^2(x) +\mplus'(x),\qquad \faminus\faminusstar = -\partial_x^2 + \mminus^2(x) +\mminus'(x).
\end{displaymath}
Both operators are equal to $-\partial_x^2+\mplus^2(x)$ from $\fH_1$ to $\fH_1^*$ (in one space dimension) up to the relatively compact perturbations $m'_\tau(x)$ and $m'_0(x)+\mplus^2(x)-\mminus^2(x)$, respectively. They are therefore homotopic as Fredholm operators (with vanishing index). They are also homotopic to $\fa_\alpha^*\fa_\alpha$ for the same reason. With the same reasoning with the operators $\fa_v$, we obtain that $\hplusv\hplusv^*$ and $\hminusv\hminusv^*$ are homotopic Fredholm operators from $\mH_1^*$ to $\mH_1$ (in two space dimensions).

This proves that $f_v$ and $g_v$ are homotopic Fredholm operators and that $D_f$ and $D_g$ are homotopic Hermitian operators. 
This shows that the operators with $\Mtop=\Ntop=1$ and $\Mtop=\Ntop=0$ are homotopic. Repeating the argument a finite number of times justifies using the index $\Mtop-\Ntop$ to characterize a large class of edge models that are homotopic to each other in the sense described above.
\end{proof}

\subsection{Mod 2 Index theory for TR edge Hamiltonians}
\label{sec:mod2}

Let us now assume that the Hamiltonians satisfy a fermionic time reversal symmetry (TRS) whose definition is given below. Consider first a Hamiltonian with $\Mtop$ propagating modes $\mhplus=\hplus^{\oplus \Mtop}$. This operator does not satisfy TRS. A time reversal symmetric operator is obtained by direct sum $\mhplus\oplus \bar \mhplus$. We may also consider the presence of $\Mtriv$ trivial edge modes $\mhminus=\hminus^{\oplus \Mtriv}$, which becomes a time reversal operator after direct summation with $\bar \mhminus$. Let us define the unperturbed operator as 
\begin{equation}\label{eq:H0TRS}
  H_0 = \left(\begin{matrix}
  \mhplus\oplus\mhminus &0\\0& \bar \mhplus\oplus\bar\mhminus
\end{matrix}\right).
\end{equation}
A general Hamiltonian $H$ satisfying the fermionic TRS is one such that 
\begin{equation}
\label{eq:TRS}
 \theta H \theta^{-1}=H,
\end{equation}
where $\theta=\mT\mK$, with $\mK$ complex conjugation and $\mT$ given for the above representation (as $2(\Mtop+\Mtriv)$ blocks) by $\mT=i\sigma_2\otimes I$. We verify that $\mT^{-1}=-\mT$ so that $\mT^2=-1$ and hence $\theta^2=-1$ as well. The $-1$ above is a characteristic of fermionic TRS and is crucial for the topological protection. Note that $\theta$ is an anti-linear transformation, such that $\theta(\alpha\psi)=\bar\alpha\theta(\psi)$ for $\alpha\in\Cm$. 

We verify that $\theta H_0\theta^{-1}=H_0$ for the above operator.  We also verify that the most general Hermitian TR preserving perturbation is of the form
\begin{equation}\label{eq:HTRS}
  H = H_0 + \left(\begin{matrix}
   V_1 & -\bar V_2 \\ V_2 & \bar V_1
\end{matrix}\right),
\end{equation}
with $V_1=V_1^*$ Hermitian and $V_2^T=-V_2$ an anti-symmetric operator; here ${}^T$ denotes symmetric transposition. In other words, $V$ is the complex representation of a matrix of quaternions. We will assume that $V_1$ and $V_2$ are appropriate perturbations of the leading term $H_0$ that satisfy the above constraints.

We define the mod 2 index \cite{PhysRevB.76.045302,1751-8121-42-36-362003,PhysRevLett.95.146802,schulz2013z_2} of $H$ as
\begin{equation}\label{eq:Z2}
  \ind_2 H = \Mtop\ {\rm mod} \ 2.
\end{equation}
We show that the above operators $H$ are indeed classified as two classes of homotopic families of Fredholm operators satisfying the TRS.

Let us first assume that $\Mtop=2$ and show that the Hamiltonian with two pairs of edge modes is homotopic to a trivial case. As in the preceding section, this requires regularizing the Hamiltonians so we can define Fredholm operators. It serves our purpose to choose a different sign of the regularization for each element in the pair. 
Let $H_0=\hplus\oplus \hplus\oplus \bar \hplus\oplus \bar \hplus$. We define the regularization
\begin{displaymath}
  H_{0v} = \hplusv \oplus \hplusv^* \oplus \bar \hplusv \oplus \bar \hplusv^* = \left(\begin{matrix}
  \hplusv \oplus \hplusv^* & 0 \\ 0 & \bar \hplusv \oplus \bar \hplusv^*
\end{matrix}\right).
\end{displaymath}
Note that $\hplusv^*$ is the same as a regularization with $v$ replaced by $-v$ (since the Hermitian conjugation changes $i$ into $-i$). 

In the general setting, we define  the regularization $(\mhplus\oplus\mhminus)_v$ as prescribed in the case $\Mtop=2$ by alternating the sign of $v$ so that $(\mhplus\oplus\mhminus)_v=\hplusv\oplus \hplusv^*\oplus \hplusv\ldots$ and then ensuring that the regularization of $\bar \mhplus\oplus\bar\mhminus$ is the complex conjugation of that of $\mhplus\oplus\mhminus$. This uniquely defines $H_{0v}$ and hence $H_v=H_{0v}+V$.

In the matrix representation \eqref{eq:HTRS}, any operator (not necessarily Hermitian as $H_v$ is no longer Hermitian) $H$ such that $\theta^{-1}H\theta=H$ is of the form
\begin{equation}\label{eq:TRS2}
  H = \left(\begin{matrix}
  H_1 & H_2 \\ H_3 & H_4
\end{matrix}\right) \quad \mbox{ for } \quad H_4 = \bar H_1,\quad H_3 = -\bar H_2.
\end{equation}
We verify that the regularization $H_v$ satisfies $\theta^{-1}H_v\theta=H_v$ and thus satisfies the TRS.

Then we have the result:

\begin{theorem}
\label{thm:topoTRS} Let $H_1$ and $H_2$ be two edge Hamiltonians satisfying the TRS constraint \eqref{eq:HTRS}, or equivalently $\theta^{-1}H_j\theta=H_j$ for $j=1,2$.  Let us assume that $\ind_2H_1=\ind_2H_2$ and let $H_{1v}$ and $H_{2v}$ be the regularizations as described above. Then there is a continuous family of Fredholm operators $F_t$ for $0\leq t\leq 1$ respecting the TRS \eqref{eq:TRS} with $F_0=H_{1v}$ and $F_1=H_{2v}$.
\end{theorem}
\begin{proof}
The proof is similar to that of Theorem \ref{thm:topo}. We highlight the differences.
For any relatively compact perturbation $V$ such that $\theta^{-1}V\theta=V$, we obtain that $H_v+V$ is homotopically equivalent to $H_v$. We also define the family of Fredholm operators
\begin{displaymath}
  \left(\begin{matrix}
  h_v & 0 \\ 0 & 1
\end{matrix}\right) R_t \left(\begin{matrix}
  0&1\\0&h_v^*
\end{matrix}\right) = \left(\begin{matrix}
 c_t h_v & s_t h_vh_v^* \\ -s_t & c_t h_v^*
\end{matrix}\right) =:H_t.
\end{displaymath}
Here $h_v$ is $\hplusv$ or $\hminusv$. This shows that $H_0$ is homotopically equivalent along a TR symmetric path ${\rm Diag}\ (H_t,\bar H_t)$ to
\begin{displaymath}
\left(\begin{matrix}
0&\hplusv \hplusv^* &0&0\\-1&0&0&0\\0&0&0&\bar \hplusv\bar \hplusv^* \\ 0&0&-1&0
\end{matrix}\right).
\end{displaymath}
As in the proof of Theorem \ref{thm:topo}, the latter is then homotopic to the case with $\hplusv$ replaced by $\hminusv$. This shows that the pair of non-trivial edge modes is continuously deformed to a pair of trivial edge modes. The generalization to arbitrary $\Mtop$ is then obvious. One has to make sure that the numbers of regularization by $v>0$ is the same as the number of regularizations by $-v<0$ except possibly for one mode pair that cannot be paired.
%
\end{proof}

%
\section{Scattering theory}
\label{sec:scat}

We now consider the scattering theory for the model edge Hamiltonians introduced in the preceding sections.  We first focus on the scattering theory for the  block $\hplus$. The general scattering theory is then obtained by direct summation as in the definition of $H_0$ in \eqref{eq:H0} or in \eqref{eq:H0TRS}. We show that the scattering theory is directly affected by the index $\ind H$. We consider the scattering theory in the presence of TRS in the next section, where we obtain that scattering also depends on the index $\ind_2 H$. 

\subsection{Spectral decomposition}

Consider the operator $\hplus$ and its partial Fourier transform (from $y$ to the dual variable $\zeta$):
\begin{displaymath}
  \hat \hplus(\zeta) = \left(\begin{matrix}
   \zeta & \faplusstar \\ \faplus & -\zeta
\end{matrix}\right).
\end{displaymath}
For each $\zeta\in\Rm$, we diagonalize the above Fredholm operator.
We denote $\fa:=\faplus$ to simplify notation and look for solutions of $\hat \hplus\psi=E\psi$ for $\psi=(\psi_1,\psi_2)^t$. We find:
\begin{displaymath}
  \zeta\psi_1+\fa^*\psi_2=E\psi_1,\quad \fa\psi_1-\zeta\psi_2=E\psi_2,
\end{displaymath}
so that $\fa^*\fa\psi_1=(E^2-\zeta^2)\psi_1$ and $\fa\fa^*\psi_2=(E^2-\zeta^2)\psi_2$. The operators $\fa^*\fa$ and $\fa\fa^*$ are Fredholm self-adjoint operators that admit the spectral decomposition
\begin{displaymath}
  \fa^*\fa \nu_k = \eps_k \nu_k,\quad \Nm \ni k\geq0,\quad \mbox{ and } \quad \fa\fa^*\mu_k=\eps_k\mu_k,\quad \Nm\ni k\geq1.
\end{displaymath}
The positive eigenvalues of $\fa^*\fa$ and $\fa\fa^*$ are necessarily the same. With our assumption on $m(x)$, we find as in, e.g., \cite{Fruchart2013779,RevModPhys.83.1057} that $\eps_0=0$ is a simple eigenvalue of $\fa$ and $\fa^*\fa$ and not of $\fa\fa^*$. We denote by $\nu_0(x)$ the normalized solution of $\fa\nu_0=0$, which is of the form $\nu_0(x)=\alpha e^{-\int_0^x \mplus(t)dt}$. We define $\eps_k=\eta_k^2$ for $k\geq0$.
We normalize $\|\nu_k\|_{L^2(\Rm)}=1$ for $k\geq0$ and then $\mu_k=\eta_k^{-1}\fa\nu_k$ for $k\geq1$ so that $\eta_k\nu_k=\fa^*\mu_k$ as well. Note that both families $(\nu_k)_{k\geq0}$ and $(\mu_k)_{k\geq1}$ are then complete families in $L^2(\Rm_x)$. This gives the eigenvalues
\begin{displaymath}
  E_{k\pm}(\zeta) = \pm \sqrt{\eps_k+\zeta^2},
\end{displaymath}
for $k\geq1$ while $E_0(\zeta)=\zeta$.

For $k=0$, we find the mode:
\begin{displaymath}
 \Big[ E_0(\zeta)=\zeta,\quad \phi_0(x,\zeta) = \left(\begin{matrix}
  \nu_0(x) \\ 0
\end{matrix}\right)\Big].
\end{displaymath}
For $k\geq1$ and the notation $l=(k,\pm)$, we find the modes:
\begin{displaymath}
  \Big[ E_l(\zeta) = \pm \sqrt{\eps_k+\zeta^2},\quad \phi_l(x,\zeta) = \dfrac1{\sqrt{2E_l(E_l+\zeta)}}\left(\begin{matrix}
  (E_l+\zeta)\nu_k(x) \\ \eta_k \mu_k(x) 
\end{matrix}\right) \Big].
\end{displaymath}
Note that $E_l(E_l+\zeta)>0$ for $k\geq1$. For concreteness, the dispersion relation of the first few modes when $m(x)=5x$ is sketched in Fig. \ref{fig:dr}.
\begin{figure}
\center
\includegraphics[width=6cm]{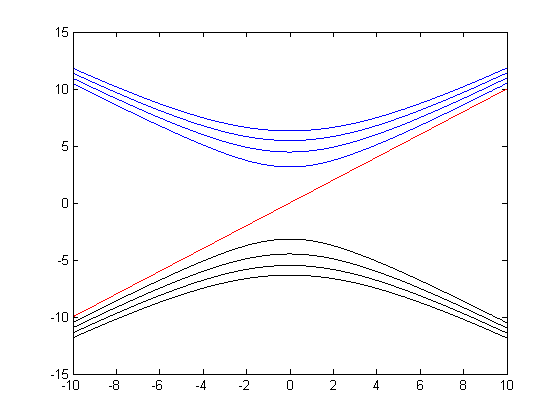}
\caption{Dispersion relation $E_l=E_l(\zeta)$ for $-3\leq l\leq 3$ and $m(x)=5x$.}
\label{fig:dr}
\end{figure}

The general solution of the equation $i\partial_t u = \hplus u$ is therefore
\begin{displaymath}
  u(t,x,y) = \dint_{\Rm}\dsum_l e^{-iE_l(\zeta)t} e^{i\zeta y} \phi_l(x,\zeta) \hat u_l(\zeta)d\zeta,
\end{displaymath}
for $\hat u_l(\zeta)$ uniquely characterized by the initial conditions $u(0,x,y)$  for instance.

\subsection{Waveguide decomposition}

The preceding spectral decomposition was obtained for each fixed $\zeta$. Scattering theory aims at fixing the energy $E$ and considering the compatible modes $\zeta=\zeta(E)$.  Let therefore $E$ be fixed and such that $E^2\not=\eta_k^2=\eps_k$ for $k\geq1$. The dispersion relation is $E^2=\eta_k^2+\zeta^2$ for $k\geq1$, which implies
\begin{displaymath}
  \zeta_l(E) = \pm \sqrt{E^2-\eta_k^2},\quad E^2>\eta_k^2,\qquad \zeta_l(E) = \pm i \sqrt{\eta_k^2-E^2},\quad E^2<\eta_k^2,
\end{displaymath}
again using the notation $l=(k,\pm)$ as well as the solution $\zeta_0(E)=E$.

This gives the presence of a protected edge mode $\zeta_0(E)=E$ without any time reversal (TR) mode at $-E$, a finite number of propagating modes for $\eta_k^2< E^2$ coming in pairs of TR symmetric modes, and an infinite number of evanescent modes for $\eta_k^2> E^2$ also coming in TR pairs. 

We then find the following set of modes. For $k=0$, we have
\begin{displaymath}
  \varphi_0(x,y) = \phi_0(x) e^{iEy},\quad \phi_0(x) = \left(\begin{matrix}
  \nu_0(x)\\0
\end{matrix}\right).
\end{displaymath}
For the $k\geq1$ propagating modes, we have the two linearly independent (but not orthogonal for the usual inner product) solutions
\begin{displaymath}
  \varphi_{k+}(x,y) = \phi_{k+}(x) e^{i\zeta_k y},\quad \phi_{k+}(x) = \frac1{\sqrt{j_k}} \left(\begin{matrix}
  c_k \nu_k(x) \\ s_k\mu_k(x)
\end{matrix}\right),
\end{displaymath}
and
\begin{displaymath}
  \varphi_{k-}(x,y) = \phi_{k-}(x) e^{-i\zeta_k y},\quad \phi_{k-}(x) = \frac1{\sqrt{j_k}} \left(\begin{matrix}
  s_k \nu_k(x) \\ c_k\mu_k(x)
\end{matrix}\right),
\end{displaymath}
where we have defined
\begin{displaymath}
  c_k = \frac{E+\zeta_k}{\sqrt{2E(E+\zeta_k)}},\quad s_k=\frac{\eta_k}{\sqrt{2E(E+\zeta_k)}},\quad j_k=c_k^2-s_k^2=\frac{\zeta_k}E\  (>0).
\end{displaymath}
We verify that $c_k^2+s_k^2=1$ and $c_k^2-s_k^2=j_k$. The above normalization of the vectors is such that the current\footnote{In the Heisenberg picture, the position operator $y$ evolves according to $j_y:=\dot y=i[\hplus,y]=\sigma_3$, which is the current operator along the edge for the Hamiltonian $\hplus$ in \eqref{eq:h+}. It is convenient to normalize propagating modes so that they have unit current.} in the $y$ direction is normalized: $(\phi_l,\sigma_3 \phi_l)=\epsilon_j$, (with the notation $l=(k,\epsilon_l=\pm1)$) where the inner product is that of $L^2(\Rm_x;\Cm^2)$. We also find that $(\phi_l,\phi_l)=j_k^{-1}$.

Finally, the evanescent modes are given by 
\begin{displaymath}
  \varphi_{k\pm}(x,y) = \phi_{k\pm}(x) e^{\pm|\zeta_k|y},\quad \phi_{k\pm}(x) = \frac{1}{\sqrt2} \left(\begin{matrix}
  \theta_k^{\pm1} \nu_x(x) \\ \mu_k(x)
\end{matrix}\right)  ,\quad \theta_k^{\pm1} = \frac{E\pm i|\zeta_k|}{\eta_k} \in \Sm^1.
\end{displaymath}
For the evanescent modes, we verify that the current in the $y$ direction vanishes with $(\phi_l,\sigma_3 \phi_l)=0$ while $(\phi_l,\phi_l)=1$.

\subsection{Mode decomposition and current conservation}

The above propagating and evanescent modes are the solutions to the fixed energy problem $H_0 u(x,y) = Eu(x,y)$, where $H_0=\hplus$ here. The solutions $\phi_l(x;E)$ form a basis of $L^2(\Rm_x; \Cm^2)$ so that for $H=H_0+V$ with $V$ a perturbation, we may write the solutions of $Hu=Eu$ as 
\begin{displaymath}
  u(x,y) = \dsum_l \alpha_l(y) \varphi_l(x,y).
\end{displaymath}
The objective of the scattering theory is to find equations for and analyze the behavior of the amplitudes $\alpha_l(y)$. Plugging the above ansatz into $Hu=Eu$ recalling that $H_0\psi_l=E\psi_l$, we find after some algebra that
\begin{displaymath}
  \dsum_l -i\alpha_l'(y) \sigma_3 \varphi_l  + V \alpha_l \varphi_l =0.
\end{displaymath}
By taking appropriate inner products, we find that 
\begin{displaymath}
  \epsilon_l \alpha'_l + i \dsum_{l'} \alpha_{l'} (\varphi_l,V\varphi_{l'}) =0,\qquad \epsilon_l = \pm 1 \mbox{ for } l=(k,\pm),
\end{displaymath}
for propagating modes while for evanescent modes, we find
\begin{displaymath}
   \alpha_{k\mp} ' \frac{\theta^{\pm 2}_k-1}2 + i \dsum_{l'} \alpha_{l'} (\varphi_l,V\varphi_{l'}) =0.
\end{displaymath}

The evanescent modes play a quantitative role in the propagation of all modes but qualitatively do not modify the regime of propagation  \cite{FGPS-07}. Their inclusion involves significant technical complications that we do not consider here. We will therefore neglect them by making the following:
\begin{hypothesis}
\label{hyp:prop} We consider perturbations of the form $V=\sum_{lm} \tV_{lm}(y) j_l^{-1}j_m^{-1}\phi_l\otimes\phi_m$, where the summation is done over the propagating modes only. 
\end{hypothesis}
Such a perturbation is local in $y$ but not local in $x$ since its integral (Schwartz) kernel is given by $\sum_{lm} \tV_{lm}(y) j_l^{-1}j_m^{-1}\phi_l(x')\otimes \phi_m(x)$. Assuming that $\tV_{lm}(y)$ is bounded, we obtain that $V$ is a relatively compact perturbation of $H_{0v}$ as introduced in the preceding section. As a consequence, $H_{0v}+tV$ for $t\in[0,1]$ provides a continuous path of Fredholm operators linking $H_{0v}$ and $H_v$.

With this convenient simplifying assumption, we obtain the system of equations
\begin{equation}\label{eq:alpha}
  \epsilon_l \alpha_l' (y) + i \dsum_{m} e^{i(\zeta_m-\zeta_{l})y} \tV_{lm}(y) \alpha_{m}(y) =0,
\end{equation}
where the summation is understood over all propagating modes $1\leq m\leq \mn$ with $\mn=\mn(E)$ an odd number. Since $V$ is Hermitian, we find that $\bar \tV_{lm}=\tV_{ml}$ in the above expression. Considering the complex conjugate equation, we obtain that
\begin{displaymath}
  \epsilon_l(\bar \alpha_l\alpha_l'+\alpha_l\bar\alpha_l') + i\sum_m [e^{i(\zeta_m-\zeta_{l})y}\tV_{lm}\bar\alpha_l\alpha_m - e^{i(\zeta_l-\zeta_{m})y}\bar \tV_{lm}\bar\alpha_m\alpha_l]=0.
\end{displaymath}
Summing over $j$ and using $\bar \tV_{lm}=\tV_{ml}$, we find the current conservation
\begin{displaymath}
  \dsum_l \epsilon_l |\alpha_l|^2(y) \quad \mbox{ is constant in $y$}.
\end{displaymath}
On a given interval $(0,L)$, this translates into
\begin{equation}\label{eq:currentconservation}
  \dsum_k |\alpha_{k+}|^2(L) + |\alpha_{k-}|^2(0) = \dsum_k |\alpha_{k-}|^2(L) + |\alpha_{k+}|^2(0).
\end{equation}
Here, the summation is over all $\mn$ propagating modes knowing that there is one $0+$ mode but  no $0-$ mode (so that $\alpha_{0-}(y)\equiv0$, say).

\subsection{Summary on propagating modes}
\label{sec:summary}

To sum up the preceding derivation, we constructed a finite number of unperturbed propagating modes $\varphi_l(x,y)$ carried by the Hamiltonian $\hplus$. Such modes are characterized by a current $\epsilon_l j_l$ with $\epsilon_l=\pm1$ describing the direction of propagation of the mode. They are of the form 
\begin{equation}
\label{eq:propmode}
  \varphi_l(x,y) = \phi_l(x) e^{i\zeta_l y}
\end{equation}
with $\zeta_l\in\Rm$ a phase velocity with the same sign as $\epsilon_l$. The transverse component $\phi_l(x)$ is normalized so its current in the $y$ direction $(\phi_l,\sigma_3\phi_l)=\epsilon_l$.

In the decomposition of $\hplus$, we find one mode $l=0$ such that $\zeta_0=E$ (hereafter called the zero mode) and pairs of modes $l=k\pm $ with the same phase $|\zeta_l|$ and current $j_l$ and  propagating in opposite directions $\epsilon_{k\pm }=\pm1$ (hereafter called pairs of non-zero modes). There are therefore $\mnplus$ propagating modes for $\mnplus$ an odd number.

The decomposition of $-\hplus$, or equivalently $\bar \hplus$, provides a similar decomposition with one mode associated with a phase velocity equal to $-E$ and pairs of modes with opposite phase velocities and currents.

We can similarly consider the decomposition of the Hamiltonian $\hminus$. The only difference with respect to the decomposition of $\hplus$ is that no mode $l=0$ exists as $0$ is not an eigenvalue of either $\faminus$ or $\faminusstar$. As a consequence, we find $\mnminus$ propagating modes of the form \eqref{eq:propmode} coming in pairs, where $\mnminus$ is now an even number.

This provides a full mode decomposition for the operator $H_0$ in \eqref{eq:H0}, with a total of propagating modes equal to $\mn_{{\cal N}}=(\Mtop+\Ntop)\mnplus+(\Mtriv+\Ntriv)\mnminus$.

%
\subsection{Scattering framework}

We focus here on the propagating modes of the operator $\hplus$ and set $\mn=\mnplus$ to simplify notation. All scattering operators are obtained under the coupling assumption of Hypothesis \ref{hyp:prop}.

Let us consider first the case $E^2<\eps_1$ so that there is only one propagating mode $\mn=1$. In that case, we find that 
\begin{displaymath}
  \alpha_0'(y) + i \tV(y) \alpha_0(y) =0,
\end{displaymath}
so that $\alpha_0(y)$ is given by the exponential of the integral of $i\tV(y)$. The random perturbation $\tV(y)$ translates into a mere phase shift and $|\alpha_0|^2(y)$ is constant. This corresponds to a setting with a perfect transmission and no backscattering. It is the mode that is typically referred to as topologically protected \cite{RevModPhys.82.3045,RevModPhys.83.1057} since its transport properties are not affected by the presence of the random perturbation $\tV(y)$. Experimental evidence of such protection may be found, e.g., in \cite{khanikaev2013photonic,wang2009observation}. As the rest of the paper shows, this total absence of backscattering is related to the energy constraint $E^2<\eps_1$ and not only to the non-trivial topology.

Let us consider the topologically equivalent but practically less favorable case where $\eps_1<E^2<\eps_2$ so that we have three propagating modes, $\alpha_0$ and $\alpha_{\pm1}$. Current conservation implies that 
\begin{displaymath}
  |\alpha_0|^2(y) + |\alpha_1|^2(y) - |\alpha_{-1}|^2(y) \quad \mbox{ is conserved}.
\end{displaymath}
Even though the index of the regularization of the operator $\hplus$ is equal to $1$, the zero mode undergoes scattering and couples with the modes $\alpha_{\pm1}$. The scattering properties of such a problem are, however, very much affected by the non-trivial topology as we shall see below.

More generally, let us assume that there are $\mn$ propagating modes and let $\alpha$ be the vector of amplitudes $(\alpha_l)_{1\leq l\leq\mn}$, where the amplitudes with positive $\epsilon_l$ come first and the amplitudes with negative $\epsilon_l$ last. We assume that the amplitude of the zero mode $\alpha_0$ comes at the "center" of the vector $\alpha$. We recast the equation \eqref{eq:alpha} for $\alpha$ as
\begin{equation} \label{eq:NU}
  \epsilon \alpha' + i \mV \alpha=0,\qquad \mV=(\mV_{mn})_{m,n},\quad \mV_{mn}(y)=e^{-i(\zeta_m-\zeta_n)y} \tV_{mn}(y),
\end{equation}
for a Hermitian matrix $\mV$ of local multipliers with $\bar \tV_{mn}(y)=\tV_{nm}(y)$ and a diagonal matrix $\epsilon={\rm Diag}(\epsilon_l)$. For $\mn=2k+1$ the size of $\alpha$, we find that $\epsilon$ is the diagonal matrix with $k+1$ times $1$ and $k$ times $-1$ on its diagonal. There are $\mn$ linearly independent solutions to the above ODE, which we combine into the $\mn\times \mn$ matrix $P(y)$ such that $P(0)=I_{\mn}$ the $\mn\times \mn$ identity matrix. We then find that 
\begin{equation}\label{eq:prop}
  \epsilon P' + i\mV P=0,\qquad P(0)=I_\mn.
\end{equation}
$P$ is the propagation or transfer matrix. Note that $\alpha(y)=P(y)\alpha(0)$ as a defining property of $P(y)$.

Consider a slab given by $y\in[0,L]$ for some $L>0$. Let us now define the central element in the theory, namely the scattering matrix $S$
\begin{displaymath}
  S = \left(\begin{matrix}
  R_+ & T_- \\ T_+ & R_-
\end{matrix}\right),
\end{displaymath}
where $R_+$ is the $k\times (k+1)$ matrix of reflection of the modes from bottom (negative values of $y$) to bottom, $T_+$ is the $(k+1)\times(k+1)$ transmission matrix of the same modes to the top (positive values of $y$), $R_-$ is the $(k+1)\times k$ matrix of reflection of the modes from top to top and $T_-$ is the $k\times k$ matrix of transmission of those modes to the bottom. 

Let us consider the first column $\alpha$ of the scattering matrix. This corresponds to $\alpha(0)$ given by $1$ followed by $k$ zeros followed by $k$ reflection coefficient $r_j$ while $\alpha(L)$ is given by $k+1$ transmission coefficients $t_j$ followed by $k$ zeros. Conservation of current implies that
\begin{displaymath}
 1 - \dsum_j|r_j|^2 = \dsum_j |t_j|^2 \quad \mbox{ so that }\quad \dsum_j |S_{1j}|^2 =1.
\end{displaymath}
This holds for any column of $S$ as well as by linearity of the above equation for $P$ for any linear combination of columns of $S$. From this, we deduce the standard property that $S$ is a unitary matrix so that $S^{-1}=S^*$. This implies the relations
\begin{displaymath}
  R_+^*R_++T_+^*T_+ = I_{k+1},\quad R_-^*R_-+T_-^*T_- = I_{k},\quad R_+^*T_-+T_+^*R_-=0,\quad T_-^*R_++R_-^*T_+=0.
\end{displaymath}
We now state the main result of this section:
\begin{theorem}
\label{thm:scatTRB}
Consider first the setting $\Mtop=1$ and $\Ntop=0$. Let $E^2\not=\eps_k$ for $k\geq1$, and $\alpha$ be the vector of mode amplitudes with the first $k+1$ components propagating toward increasing values of $y$ and the last $k$ components toward decreasing values of $y$. Consider a random slab of thickness $L$, i.e., $V(y)=0$ outside of $[0,L]$.

Then there is a (non-trivial) incoming vector $\tau=(\tau_j)_{1\leq j\leq k+1}$ that is purely transmitting, i.e., such that $R_+\tau=0$. 
Moreover, the conductance satisfies ${\rm Tr} (T_+^*T_+)\geq1$.

For the general case of $H_0$ in section \ref{sec:summary}, with $\mnplus=2k_++1$ and $\mnminus=2k_-$, we have $(\Mtop+\Ntop)k_++(\Mtriv+\Ntriv)k_-$ modes propagating in each direction with an extra $\Mtop$ modes propagating toward increasing values of $y$ and an  extra $\Ntop$ modes propagating toward decreasing values of $y$. Then, there are  $\ind H_v=\Mtop-\Ntop$ (non-trivial) linearly independent transmission vectors $\tau^l$ that are transmitted without reflection ($R_+\tau^l=0$) and the conductance satisfies ${\rm Tr} (T_+^*T_+)\geq \Mtop-\Ntop$.
\end{theorem}
\begin{proof}
Since $R_+$ is a $k\times (k+1)$ matrix, it is at most of rank $k$. Let $R_j$ be the $1\leq j\leq k$ rows of $R_+$. Then there exists (at least) one (normalized) vector $\tau$ orthogonal to all of them. This vector satisfies $R_+\tau=0$ by construction. The reflection matrix $R_+$ can be obtained experimentally by means of $k$ scattering experiments so that $\tau$ is observable. Note that it depends on $V$ and hence is fluctuation dependent.

Recall the conservation $R_+^*R_++T_+^*T_+ = I_{k+1}$.  Since $R_+$ is of rank $k$ at most, we deduce that $1$ is necessarily an eigenvalue of the matrix $T_+^*T_+$.  As a consequence, there is necessarily transmission thanks to the topological protection, and physically, the Landauer formula \cite[(33)]{RevModPhys.69.731} implies that conductance is at least equal to $1$ (in those units). 

When the index is $\Mtop-\Ntop$, then the above picture generalizes to a setting with $\mn_{\cal N}$ propagating modes. There is an excess of $\Mtop-\Ntop$ modes propagating toward positive values of $y$ compared to those propagating in the opposite direction. As a consequence, we find that $T_+^*T_+$ admits at least $\Mtop-\Ntop$ eigenvalues equal to $1$ so that the conductance in the Landauer formula is at least equal to $\Mtop-\Ntop$. Moreover, we now find that $\Mtop-\Ntop$ (non-trivial, linearly independent) transmission vectors $\tau^l$ may be transmitted without reflection, which are the vectors orthogonal to the all the rows of $R_+$. 
This completes the proof of the result.
\end{proof}

This result shows that complete localization \cite{RevModPhys.69.731,FGPS-07}, which would correspond to an exponential decay of the eigenvalues of $T_+^*T_+$ as $L$ increases, cannot happen when $\ind H_v\geq1$. This is a direct consequence of the topological protection of the edge modes.

However, this does not mean that $T_+^*T_+$ is close to identity in the presence of randomness. We will show in the case $\Mtop=1$ and $k=1$ that significant reflection occurs in the presence of random fluctuations. Note that $R_-^*R_-$ can be of maximal rank $k$ and so we do not expect the eigenvalues of $T_-^*T_-$ to be large in the presence of randomness. Indeed, when $k=1$, we obtain that $T_-^*T_-$ goes to $0$ as $L$ increases in the presence of random fluctuations. So unless $k=0$, which occurs when $E^2<\eps_1$ in our model $\Mtop=1$, there is both significant transmission as well as significant reflection.

Significant transmission, i.e., ${\rm Tr} (T_+^*T_+)\geq1$, and in fact equal to $1$ asymptotically as we shall see in section \ref{sec:diff}, is topologically protected. For reflection, the picture is as follows. When $E^2<\eps_1$, then no reflection occurs when the protected mode is transmitted. Note that when $m(x)=\lambda x$, then $\eps_1$ is proportional to $\lambda$. As a consequence, a sharp transition between two material topologies is energetically favorable for the presence of one and only one edge mode.  This is protected energetically, not topologically. When $E^2>\eps_1$, then significant reflection occurs when the unperturbed protected mode is transmitted. However, there is a random vector, which plays the role of the protected mode in the random environment and is purely transmitted. In that sense (and in that sense only), the absence of reflection is also topologically protected.

\section{Scattering theory in the TRS setting}
\label{sec:scatTR}

Consider an operator $H_0=\hplus^{\oplus \Mtop} \oplus \bar \hplus^{\oplus \Mtop}$ satisfying the TR symmetry $\theta H_0\theta^{-1}=H_0$. For concreteness, we assume $\Mtop=1$, with obvious generalizations to the case $\Mtop>1$. Recalling the solutions of $(\hplus-E)\psi=0$ in the preceding section, we obtain the following solutions for $(H_0-E)\psi=0$ given by 
\begin{displaymath}
  \varphi_l(x,y) = \left(\begin{matrix}
  \phi_l(x) \\ 0 
\end{matrix}\right) e^{i\zeta_l y}\quad \mbox{ and } \quad \theta\varphi_l = \left(\begin{matrix}
  0 \\ -\phi_l(x) 
\end{matrix}\right) e^{-i\zeta_ly}.
\end{displaymath}
Here, each $\phi_l$ is a two-vector while $\varphi_l$ and $\theta\varphi_l$ are now four-vectors. We verify that $\varphi_l$ and $\theta\varphi_l$ are orthogonal vectors associated to the same energy $E$. This is the standard Kramers degeneracy.

Indeed, for $H_0\psi=E\psi$, we have $\theta H_0\psi=\theta H_0\theta^{-1}\theta \psi = H_0\theta\psi=E\theta\psi$ so $\theta\psi$ is another eigenvector of $H_0$ with energy $E$. Now $(\psi,\theta\psi)=(\psi,\mT\bar\psi)=-(\mT\psi,\bar\psi)=-(\psi,\mT\bar\psi)=-(\psi,\theta\psi)=0$. More generally, $(\theta a,\theta b)=(a,b)$ by the same reasoning.

It is this orthogonality of time reversed modes for $\theta^2=-1$ that is responsible for the topological protection when $\Mtop=1\ \rm mod\  2$.

Let us now look for solutions of $H\psi=E\psi$ with $H=H_0+V$ and use the decomposition
\begin{displaymath}
  \psi(x,y) = \dsum_l a_l(y) \varphi_l(x,y) + b_l(y) \theta\varphi_l(x,y),\quad
  \theta\psi(x,y) = \dsum_l \bar a_l(y) \theta\varphi_l(x,y)  -\bar b_l(y) \varphi_l(x,y) .
\end{displaymath}
Here, as above, we use the anti-linearity of the map $\theta$. 

Using the same method as earlier, we find the equations for the above amplitudes
\begin{displaymath}
  \begin{array}{rcl}
   \epsilon_l a_l' &=& -i \dsum_m (\varphi_l,V\varphi_m) a_m + (\varphi_l,V\theta\varphi_m) b_m \\[3mm]
   -\epsilon_l b_l' &=& -i \dsum_m (\theta\varphi_l,V\varphi_m) a_m + (\theta\varphi_l,V\theta\varphi_m) b_m.
\end{array}
\end{displaymath}
We recast these as 
\begin{displaymath}
  \epsilon_ma_m' = -i A_{ml}a_l-i\bar B_{ml}b_l,\quad -\epsilon_mb_m' = iB_{ml}a_n-i\bar A_{ml}b_n,
\end{displaymath}
for $\mn\times \mn$ matrices $A=A^*$ and $B=-B^T$ so that $2\mn$ is the total number of propagating modes. We assume as before that the TR symmetric matrix $V$ is chosen so that only the propagating modes are coupled and are independent of the evanescent modes, which we neglect here. We then have the current conservation
\begin{displaymath}
  \Big(\sum_m \epsilon_m|a_m|^2 - \epsilon_m |b_m|^2\Big)'=0.
\end{displaymath}
With $\epsilon$ the diagonal matrix with $\epsilon_m$ as its entries (twice), we have
\begin{displaymath}
  \epsilon \left(\begin{matrix}
 a\\b
\end{matrix}\right)' = -i M \left(\begin{matrix}
  a\\b
\end{matrix}\right), \quad M =  \left(\begin{matrix}
  A & \bar B \\ B & -\bar A
\end{matrix}\right),\quad MJ=-J\bar M,\quad J=\left(\begin{matrix}
0&1\\-1&0
\end{matrix}\right).
\end{displaymath}
With $c=(a,b)^t$, we find $\theta(\epsilon c)'=i\theta(Mc)=iJ\bar M\bar c=-iMJ\bar c=-iM\theta c$ so that $c$ and $\theta c$ solve the same equation, where $\theta=J\mK$ in the above representation since $J\equiv i\sigma_2$. So, the property that the equation is TR symmetric remains true for the one-dimensional scattering equation for $c(y)$. We still use the notation $\theta$ to represent that symmetry.

Let $Q$ be the solution of the equation for $c$ with initial conditions given by $(I \ 0)^t$ for $I$ the $\mn\times \mn$ identity matrix. We then find that $-\theta Q$ is the solution with initial conditions given by $(0,I)^t$. As a consequence, the transfer matrix $P$ solution of the above equation with $I_{2\mn\times 2\mn}$ initial conditions is given by
\begin{displaymath}
  P = (Q\ \ -\theta Q), \qquad Q=\left(\begin{matrix}
\alpha\\\beta
\end{matrix}\right), \qquad P = \left(\begin{matrix}
 \alpha & -\bar\beta \\ \beta&\bar \alpha
\end{matrix}\right).
\end{displaymath}
We recover that $P$ is a complex representation of a matrix of quaternions, which is a standard representation of transfer matrices of problems with $\theta^2=-1$ TR symmetry.

Let us consider a representation where the first $N$ components of $c$ are the forward propagating modes $f(z)$ corresponding to positive values of $\zeta_l$ whereas the last $N$ components of $c$ are the backward propagating modes $b(z)$ with negative values of $\zeta_l$. In this choice of basis (obtained by a permutation from the original basis), we have 
\begin{displaymath}
  \left(\begin{matrix}
  f(L) \\ b(L) 
\end{matrix}\right) = P(L) \left(\begin{matrix}
  f(0)\\b(0)
\end{matrix}\right).
\end{displaymath}
We still use the notation $P$ for the transfer matrix in this new basis. The scattering matrix is then given by
\begin{displaymath}
  S=S(L) =\left(\begin{matrix}
  R_+ & T_- \\ T_+ & R_-
\end{matrix}\right)\quad \mbox{ so that } \quad \left(\begin{matrix}
  b(0) \\ f(L) 
\end{matrix}\right) = S \left(\begin{matrix}
  f(0) \\ b(L)
\end{matrix}\right)
\end{displaymath}
for all possible solution $c(z)=(f(z),b(z))^t$ of the above equation. The matrix $S$ is unitary as we already observed. However, thanks to the TR symmetry, it satisfies additional properties. We find
\begin{displaymath}
  c(z) = \left(\begin{matrix}
 f(z) \\ b(z)
\end{matrix}\right),\quad \theta c(z) = \left(\begin{matrix}
 \bar b(z) \\ -\bar f(z)
\end{matrix}\right)\quad \mbox{ so that } \left(\begin{matrix}
 -\bar f(0) \\ \bar b(L) 
\end{matrix}\right) = S \left(\begin{matrix}
 \bar b(0) \\ -\bar f(L)
\end{matrix}\right)
\end{displaymath}
and hence, with $\sigma_3={\rm diag(I_{\mn},-I_{\mn})}$ (i.e., really $\sigma_3\otimes I_{\mn}$)
\begin{displaymath}
  -\sigma_3 \left(\begin{matrix}
  -f(0) \\ b(L) 
\end{matrix}\right) = -\sigma_3 \bar S \sigma_3 \sigma_3 \left(\begin{matrix}
  b(0) \\ -f(L)
\end{matrix}\right)\ \mbox{ so } \  \left(\begin{matrix}
  f(0) \\ b(L)
\end{matrix}\right) = -\sigma_3 \bar S \sigma_3 \left(\begin{matrix}
  b(0) \\ f(L) 
\end{matrix}\right) = S^{-1} \left(\begin{matrix}
  b(0) \\ f(L) 
\end{matrix}\right).
\end{displaymath}
Since $S$ is unitary, we deduce that $S^{-1}=S^*=-\sigma_3 \bar S \sigma_3$ so that $S^T=-\sigma_3 S \sigma_3$. In other words, the matrix $\tilde S=S\sigma_3$ is skew symmetric: $\tilde S^T=-\tilde S$. This is the representation of $S$ as a skew-symmetric matrix in an appropriate basis; see \cite{1751-8121-41-40-405203,RevModPhys.69.731}. 
The above representation is as directly useful as the skew-symmetric one. Indeed, we find that 
\begin{displaymath}
  S^T = \left(\begin{matrix}
 R_+^T & T_+^T \\ T_-^T & R_-^T
\end{matrix}\right) = -\sigma_3 S \sigma_3 = \left(\begin{matrix}
  -R_+ & T_- \\ T_+ & -R_-
\end{matrix}\right).
\end{displaymath}
Thus, $R_+$ and $R_-$ are skew-symmetric and $T_+$ and $T_-$ are transpose to each-other.

\begin{theorem}
\label{thm:scatTRS}
   Let $E^2\not=\eps_k$ for all $k\geq1$ and let $2\mn$ be the number of propagating modes in a TRS environment with $2(\Mtop+M)$ blocks. Then, when $\ind_2 H_v= \Mtop${\rm mod} $2=1$, there is a non-trivial incoming vector $\tau_+=(\tau_j)_{1\leq j\leq\mn}$ such that $R_+\tau_+=0$. Similarly, there is a non-trivial incoming vector $\tau_-$ such that $R_-\tau_-=0$. Moreover, the conductances satisfy $\rm{Tr}(T_+^*T_+)\geq \Mtop${\rm mod} $2$ and $\rm{Tr}(T_-^*T_-)\geq \Mtop${\rm mod} $2$. 
\end{theorem}
\begin{proof}
As we noted earlier,
\begin{displaymath}
  T_+^*T_+ + R_+^*R_+ = I,
\end{displaymath}
with a symmetric expression for $R_-$ and $T_-$. The above matrices are $\mn\times \mn$ and no standard index prevents $T$ from being negligible as the thickness $L$ increases. However, the mod 2 index does prevent this from happening. Indeed, the number of modes $\mn$ is equal to $\Mtop+2M$ since all non-protected modes come in TR pairs. When $\Mtop$ is even, then $\mn$ is even as well. However, when $\Mtop$ is odd, then the skew symmetric matrix $R_+$ is a matrix with odd dimension. In such a case, its rank is at most $\mn-1$ and we again obtain that $1$ is an eigenvalue of the matrix $T_+^*T_+$. That $R_+$ is at  most of rank $\mn-1$ when $\Mtop$ mod $2=1$ comes from the diagonalization property of such matrices  \cite{1751-8121-41-40-405203,schulz2013z_2,Y-CMS-61}. In other words, localization is not possible and there exists a vector $\tau_+$ as indicated above. Results for $\tau_-$ in the opposite direction are obtained by symmetry.
\end{proof}

From the skew-symmetry of $R_+$ and $R_-$, we deduce the absence of backscattering from one mode into its time reversed alter-ego, a well-documented feature in the literature \cite{doi:10.1143/JPSJ.67.2857,PhysRevLett.95.146802}. However, as in the TRS breaking setting, while transmission is protected, an amount of back-scattering is inevitable unless the protected mode is one of the modes $\tau_\pm$ defined in the above theorem.


\section{Diffusion approximation}
\label{sec:diff}

The diffusion approximation is a macroscopic model that allows us to understand the large distance influence of (typically one-dimensional) random perturbations in differential equations. In the setting of this paper, it is based on assessing the limiting behavior of $S$ and $P$ introduced in the preceding section as the thickness $L$ of the random medium increases. The reader is referred to \cite[Chapters 6\&7]{FGPS-07} for the main results and history of the diffusion approximation, in which the reflection and transmission coefficients are modeled as stochastic diffusions. The method also bears some similarities with the random matrix theory of quantum transport \cite{RevModPhys.69.731,RevModPhys.87.1037}.

%
\subsection{Classes of transmitting or localizing systems}

In the preceding sections, we showed that edge Hamiltonians belonged to separate classes depending on the value of $\ind\ H_v = \Mtop-\Ntop$ in \eqref{eq:index} and on the value of $\ind_2 H = \Mtop\ {\rm mod} \ 2$ in \eqref{eq:Z2} in the TRS setting. For random fluctuations coupling propagating modes as described in Hypothesis \ref{hyp:prop}, which we assume also holds for the rest of the section, we showed that the conductance along the edge was bounded from below by these two indices. In the topologically non-trivial cases, we therefore obtain that transport is guaranteed no matter how strong the random fluctuations $V$ are. In other words, complete (Anderson) localization, which may be characterized by asymptotically vanishing transport (conductance) as the strength of the randomness increases \cite{FGPS-07}, is not possible. This is consistent with results obtained in \cite[Theorem 6.6.3]{prodan2016bulk} for general classes of Hamiltonians and under the assumption of a mobility edge constraint that we do not need to verify in this paper. Heuristically, we expect that non-trivial topologies generate an obstruction to (complete) localization by forcing the presence of delocalized modes; see \cite[Section 6]{prodan2016bulk} for a thorough discussion.

\medskip

We now prove a more precise result: for an appropriate form of the disorder $V$ described in detail below, the edge Hamiltonians are homotopic to a setting in which the conductance is asymptotically, as the thickness $L$ of the random domain increases, {\em exactly} equal to $\ind\ H_v = \Mtop-\Ntop$ in general and $\ind_2 H = \Mtop\ {\rm mod} \ 2$ in the TRS setting. When the latter vanish, we show that conductance vanishes exponentially as $L$ increases, which is a hallmark of wave localization in one space dimension. In the case of non-trivial indices, the physical explanation is therefore that a number of protected modes equal to the non-trivial index is allowed to propagate while everything else (Anderson) localizes. We thus have a direct sum of the space of propagating modes into one component observing free transport and an orthogonal component that fully localizes.

Returning to the question of back-scattering, we show that when $E^2>\eps_1$, the unperturbed zero modes do undergo scattering, for which we obtain a quantitative description. This implies that the back-scattering free modes, which are realization  dependent, need to be defined  as described in the preceding section.

It is quite likely that the aforementioned results hold for a much larger class of random fluctuations than the ones we are considering here. However, this would require additional unknown results on the diffusion approximation that we do not develop here.

\paragraph{Scaling and choice of coupling operator.} 
The specific random fluctuations we consider to  asymptotically obtain transport of protected modes and localization of all other modes is as follows. Recall that the mode propagation is governed by \eqref{eq:NU} and the coupling to the random fluctuations by the local multipliers $\tV_{mn}(y)$. Although this is not necessary, we assume for simplicity that $\tV_{mn}(y)=W_{mn}\tV(y)$, i.e., that all couplings see a unique scalar-valued random field $\tV(y)$.

We now need to perform two choices: find which coupling coefficients $W_{mn}$ are non-zero and then find the dependency of $\tV(y)$ in the $y$ variable that leads to the diffusion approximation. Let $0<\eps\ll1$. We assume the slab of propagation through random heterogeneities of size $L_\eps=\eps^{-\frac12}L$ (for some $L>0$) large compared to the frequencies $\zeta_m$ (of order $O(1)$) of the propagating modes. We also assume that the random fluctuations oscillate rapidly at the scale $\sqrt\eps$ and thus are of the form $\tV(\eps^{-\frac12}y)$. Upon rescaling $y\to \eps^{-\frac12}y$, we thus obtain that the propagator in \eqref{eq:NU} is given by
\begin{equation}
 	\label{eq:Veps}
 	\epsilon P_\eps' + i \mV_\eps P_\eps=0,\quad \mV_\eps = (\mV_{\eps mn})_{m,n},\quad \mV_{\eps mn}(y) = e^{-i(\zeta_m-\zeta_n)\frac y{\sqrt\eps}} W_{mn} \dfrac{1}{\sqrt\eps}  \tV(\frac y\eps). 
\end{equation}
The above regime is one of the classical asymptotic regimes leading to diffusions since for a mean-zero sufficiently mixing process $\tV$, we obtain that $\eps^{-\frac12}\tV(\eps^{-1}y)$ converges to a rescaled Brownian motion as $\eps\to0$. The above scaling, corresponding to a domain of size $\eps^{-\frac12}$ and a scale of the random fluctuations (correlation length) also of $\eps^{\frac12}$ leads to the computationally least complicated of the regimes to which diffusion approximations can be applied;  see \cite[Chapters 6]{FGPS-07}, where our $\sqrt\eps$ is denoted $\eps$. Our assumption on $\tV$ is as follows.
\begin{hypothesis}\label{hyp:random}
 The process $\tV(y)$ is a mean-zero, stationary (homogeneous), ergodic process such that its infinitesimal generator ${\cal L}_\tV$ satisfies the Fredholm alternative; i.e., for any bounded function $f$ that is centered, i.e., $\E\{f(\tV(y))\}=\E\{f(\tV(0))\}=0$, then the problem ${\cal L}_{\tV}u(y)=f(y)$ admits a unique bounded solution with $\E\{u(\tV(0))\}=0$.
\end{hypothesis}
For instance, any (stationary) Ornstein-Uhlenbeck process satisfies the above constraint; see \cite[Section 6.3.3]{FGPS-07}.

\medskip

It remains to choose the coupling coefficients of the operator $V$, given above by the scalar terms $W_{mn}$, that are non-vanishing. This is done as follows. For the Hamiltonian in \eqref{eq:edgeH} with unperturbed Hamiltonian given in \eqref{eq:H0}, and following the results of section \ref{sec:scat}, we obtain $\mnplus$ propagating modes for each block $\hplus$ or $-\hplus$ (or equivalently $\bar \hplus$) and $\mnminus$ propagating modes for eack block $\hminus$ and $-\hminus$ (or equivalently $\bar \hminus$). Note that $\mnplus$ is odd while $\mnminus$ is even. Now we choose $W_{mn}$ to couple the propagating modes into $2\times2$ or $3\times3$ systems of equations as follows. 

Let $\varphi_1$ and $\varphi_2$ be two such (unperturbed) modes with respective currents $\epsilon_1j_1$ and $\epsilon_2j_2$. Then a coupling between the two modes is obtained by defining $0<\sqrt{\gamma_{ij}}=W_{ij}$ for $1\leq i,j\leq 2$ with $\gamma_{12}=\gamma_{21}$ and $W_{1j}=W_{2j}=W_{j1}= W_{j2} =0$ otherwise. It will be convenient to use the coefficients $\gamma_{ij}$ in the following. The coefficients $W_{ij}$ need not be positive but since only their square appears in the diffusive limit, we assume the above form for concreteness.
%
%

Let $\alpha_1$ and $\alpha_2$ be the corresponding amplitudes of these modes. They satisfy the coupled system, with $\tV_\eps(y) = \frac{1}{\sqrt\eps} \tV(\frac y\eps)$,
\begin{equation}\label{eq:2by2}
  \left(\begin{matrix}
  \epsilon_1\alpha_1 \\ \epsilon_2\alpha_2
\end{matrix}\right) '+ i\tV_\eps(y)\left(\begin{matrix}
  \sqrt{\gamma_{11}} & \sqrt{\gamma_{12}} e^{i(\zeta_1-\zeta_2)\frac{y}{\sqrt\eps}} \\
  \sqrt{\gamma_{12}} e^{i(\zeta_2-\zeta_1)\frac{y}{\sqrt\eps}} & \sqrt{\gamma_{22}}
\end{matrix}\right)\left(\begin{matrix}
  \alpha_1\\\alpha_2
\end{matrix}\right)=0 .
\end{equation}

The generalization to the coupling of $3$ modes $\varphi_j$ for $1\leq j\leq 3$ (or more modes, but for which we do not know how to analyze the diffusion approximation) is then straightforward, 
with a choice of coefficients $W_{ij}$ non-vanishing only for both $1\leq i,j\leq 3$ with corresponding $\gamma_{ij}=\gamma_{ji}$.
The above coupled system of equations for $(\alpha_j)_{1\leq j\leq 3}$ is then modified accordingly.

\medskip

\noindent {\bf General case without TRS.} Let $H=H_0+V$ be a general Hamiltonian as defined in \eqref{eq:edgeH} and let the energy $E$ be such that the Hamiltonian carries $\mn_{\cal N}=(\Mtop+\Ntop)\mnplus+(\Mtriv+\Ntriv)\mnminus$ propagating modes. We couple the modes (choose the coefficients $W_{mn}$) using the above $2\times2$ or $3\times3$ systems as follows. Without loss of generality, we assume $\Mtop\geq \Ntop$. For $\Ntop\geq1$, we pair $\Ntop$ zero modes of an operator $\hplus$ with that of an operator $\hminus$. For each of the $\Mtop-\Ntop$ remaining zero modes of the operators $\hplus$, we couple them with a pair of non-zero modes (with currents $\pm j_k$) if such modes are available. Each remaining pair of non-zero modes are finally coupled to form their own $2\times2$ systems.

So, our choice of Hermitian perturbation $V$ is such that each zero mode of $-\hplus$ is coupled with a  zero mode of a $\hplus$ operator. Each of the remaining $\Mtop-\Ntop$ zero modes is either part of a $3\times3$ system  if enough pairs of non-zero modes are available or solves a decoupled $1\times1$ system otherwise. Each remaining pair of non-zero modes finally solves a $2\times2$ system of equations. We then have the following result:
\begin{theorem}
\label{thm:diff}
Let $H$ be defined as in \eqref{eq:edgeH}, the perturbation $V=V_\eps$ as described above and the random process $\tV(y)$ satisfying hypothesis \ref{hyp:random}.  

Then, asymptotically as $\eps\to0$, the limiting scattering matrix $S=S(L)$ is such that ${\rm Tr}\ T_+^*T_+(L)$ converges to $\Mtop-\Ntop$ as $L\to\infty$. Moreover, let $\alpha_0$ be the amplitude of a zero-mode in a $3\times3$ system as described above. Then the resulting scattering matrix involves a reflection matrix $R_+=(r_1,r_2)$ such that $|r_1|^2+|r_2|^2$ converges to $1$ as $L\to\infty$. Finally, for some realizations of $\tV(y)$, $|r_1|^2$ and $|r_2|^2$ are as close to $1$ as one wishes (though not both at the same time).
\end{theorem}
\begin{remark}
\label{rem:symgamma} The proof of the above theorem is given under the additional assumption that each $3\times3$ coupling matrix satisfies the additional assumption that $\gamma_{12}=\gamma_{13}$. A very similar proof works when $\gamma_{12}\not=\gamma_{13}$ but requires results for parabolic equations with degenerate coefficients at the domain's boundary that likely hold based on existing theories in \cite{EM-PUP-13,FGPS-07} but nonetheless not done in detail. See a footnote in the proof for additional details.
\end{remark}

The proof of the theorem will be given in sections \ref{sec:3by3} and \ref{sec:2by2} where the asymptotic analysis of the $3\times3$ and $2\times2$ systems are presented, respectively.

\medskip

\noindent {\bf Case where TRS holds.} Let us now assume that the Hamiltonian $H$ satisfies the TRS. We thus have $\Mtop=\Ntop$ and $\Mtriv=\Ntriv$. The main difference with respect to the preceding case is that $V$ now needs to satisfy the TRS constraints indicated in \eqref{eq:HTRS}. Let us consider the modes $\varphi$ of an operator $\hplus$ and those $\theta\varphi$ of its conjugate $\bar \hplus$. Such modes cannot be coupled by random perturbations since they are orthogonal by TRS independently of the random fluctuations. The coupling $V$ thus needs to be modified accordingly.

We write $\Mtop=2p+\eta$ with $p\geq0$ and $\eta=0$ or $\eta=1$. For $p\geq1$, consider a pair of operators $\hplus$ and their complex conjugate $\bar \hplus$. Restricted to these sub-blocks, the Hamiltonian takes the form
\begin{displaymath}
   \left(\begin{matrix}
   \hplus+V_{11} & V_{12}^* & -\bar V_{13} & V_{14}^* \\
   V_{12} & \hplus+V_{22} &  - \bar V_{14} & -\bar V_{24}\\ V_{13} & -V_{14}^T & \bar \hplus + \bar V_{11} & \bar V_{12}^* \\
   V_{14} & V_{24} &  \bar V_{12} & \bar \hplus+\bar V_{22}
\end{matrix}\right).
\end{displaymath}
Here, $V_{ij}$ for $1\leq i,j\leq2$ are arbitrary Hermitian, while $V_{13}=-V_{13}^T$ and $V_{24}=-V_{24}^T$ and we verify that the constraint \eqref{eq:HTRS} implies that $V_{23}=-V_{14}^T$ as indicated above for $V_{14}$ an arbitrary Hermitian $2\times2$ matrix. The coupling between the zero mode of the first $\hplus$ with that of the first $\bar \hplus$ is therefore prohibited since the $(1,1)$ entry of $V_{13}$ needs to vanish. We already know this as these zero modes are orthogonal independently of $V$. We therefore set $V_{13}=V_{24}=0$. 

Since $V_{14}$ is arbitrary, however, we choose it so that the zero mode of the first $\hplus$ and the zero mode of the second $\bar \hplus$ are coupled. Note that the zero mode of the second $\hplus$ is now coupled to the zero mode of the first $\bar \hplus$ with the same coupling constant up to a sign, which has no influence on the final result. It remains to choose $V_{11}$ and $V_{22}$ such that all the non-zero modes of each of these blocks are coupled as we did in the TRS breaking setting.

When $\eta=0$, all propagating modes have been coupled into $2\times2$ systems of the form \eqref{eq:2by2}. When $\eta=1$, the remaining block may be obtained by retaining the first and third rows and columns in the above matrix. We then set $V_{13}$ to $0$ and are left with two uncoupled blocks. The fluctuations $V_{11}$ are then chosen as in the TRS breaking case: unless $\mnplus=1$ for the block, the zero-mode is coupled with an other pair of non-zero modes to form a $3\times3$ system. Each remaining pair of non-zero modes is coupled into a $2\times2$ system. Then we have:
\begin{theorem}
\label{thm:diffTRS} Let $H$ be a TRS Hamiltonian and $V=V_\eps$ chosen as above. Then the conclusions of Theorem \ref{thm:diff} hold with $\Mtop-\Ntop$ replaced by $\Mtop$ mod $2$.
\end{theorem}
The proof of the theorem is the same as that of Theorem \ref{thm:diff} as it involves the same $3\times3$ and $2\times2$ coupled systems. We now turn to the description of the diffusion approximation and a proof of the above theorems.

\subsection{Diffusion equations for transfer matrices}

Let us recall the $\eps-$dependent equation \eqref{eq:Veps} for the propagator
\[
  \Lambda P_\eps' + i {\cal V}_\eps P_\eps =0
\]
where we use the notation $\Lambda=\epsilon^{-1}=\epsilon$ with diagonal entries $\lambda_k$. According to the choice of coefficients $W_{mn}$ described in the preceding section, the above equation takes the form of a block diagonal equation with blocks of size $2\times2$ or $3\times3$. More generally at this level of the analysis, the transfer matrix $P_\eps(y)$ is an $\mn\times\mn$ matrix.

%

The objective of the diffusion approximation is to find the asymptotic law as $\eps\to0$ of the variables $P=Q+iR$ (both the real part $Q$ and imaginary part $R$). 
The choice of scaling presented in the preceding section 
\begin{displaymath}
  \mV_{\eps mn}(y) = \tV_\eps(y) e^{-i\frac{\zeta_{m}-\zeta_n}{\sqrt\eps} y} W_{mn}=: \frac{1}{\sqrt\eps} \tV(\frac y\eps) M_{mn}(\frac y{\sqrt\eps}).
\end{displaymath}
The multiplier $M=(M_{mn})_{m,n}$ therefore now encodes the effects of the phases $\zeta_m$ and the choice of coupling coefficients $W_{mn}$. Let $U$ be a column vector of all the variables in $Q$ and $R$. 

The diffusion operator (infinitesimal generator) in the variables $U$ is then given by \cite[comments below Theorem 6.4]{FGPS-07}
\begin{displaymath}
  \mL = \lim\limits_{T\to\infty} \dfrac1T\dint_0^T\dint_0^\infty \E [\tV0) \tilde M(\tau)U\cdot\nabla_U][\tV(y) \tilde M(\tau) U\cdot\nabla_U] dy d\tau.
\end{displaymath}
This non-trivial result is one of the main places where we use the theory explained in great detail in \cite[Chapter 6\&7]{FGPS-07}.
Here $\tilde M$ is the real-valued matrix defined such $\tilde M U$ is the column representing the real and imaginary components of the complex object $MP$. 

The above operator $\mL$ admits the following interpretation. Let $\pi(y,U',U)dU$ be the probability that $U'$ at $y=0$ becomes $U$ within $dU$ at $y>0$. Thus, $\pi(0,U',U)=\delta(U-U')$. Then we obtain  (see, e.g., \cite[Chapter 6]{FGPS-07}) that $\partial_y\pi=\mL\pi$ as a function of $(y,U')$ for a fixed $U$ while $\partial_y\pi=\mL^*\pi$ as a function of $(y,U)$ for a given initial condition $U'$ at $y=0$, where $\mL^*$ is the formal adjoint to $\mL$. These are the forward and backward Kolmogorov equations, respectively, of the limit as $\eps\to0$ of the process $P_\eps$ solution of \eqref{eq:Veps}. They indicate, in the limit $\eps\to0$, how the law of the coefficients of the propagator $P$ vary with $L$. It is in this limit that we analyze the scattering coefficients described in Theorems \ref{thm:diff} and \ref{thm:diffTRS}.

We assume that all the coefficients $\zeta_{mn}=\zeta_m-\zeta_n$ are incommensurate in order to avoid coupling between the frequencies.  Rather than working with the variables $U=(Q,R)$, we introduce the change of variables $P=Q+iR$ and $\bar P=Q-iR$ as an equivalent basis to represent the complex object $P$; see \cite[Section 20.3]{FGPS-07}. We denote by $\hat R_0=\int_0^\infty \E \{\tV(0)\tV(y)\} dy$ the power of the random fluctuations. After some lengthy calculations similar to those in \cite{FGPS-07}, we obtain the following expression
\begin{equation}\label{eq:generator}
  \mL = \hat R_0 \bar X_0 X_0 + 2\hat R_0 \dsum_{k<l} |W_{kl}|^2 (\bar X_{kl}X_{kl}+X_{kl}\bar X_{kl})
\end{equation}
where 
\begin{displaymath}
  X_0 = \dsum_{k,l} \lambda_k W_{kk} (P_{kl}\partial_{P_{kl}} - \bar P_{kl}\partial_{\bar P_{kl}})
\end{displaymath}
and
\begin{displaymath}
  X_{kl} = \lambda_k Y_{kl} - \lambda_l \bar Y_{lk} ,\qquad Y_{kl} = \dsum_m P_{lm}\partial_{P_{km}}.
\end{displaymath}

We observe that $\mL$ is a second-order differential operator. We have $\frac12\mn(\mn-1)+1$ vector fields $X_0$ and $X_{kl}$ when $P$ is an $\mn\times \mn$ operator for $\mn^2$ (complex) scalar variables.  $\mL$ is therefore a (very) degenerate operator (nowhere elliptic). 

\subsection{Diffusion approximation for $3\times3$ systems}
\label{sec:3by3}

The topological phases described in earlier sections are characterized by two types of asymptotic results, one in which localization prevails and that can be demonstrated by a $2\times2$ system as has already appeared in the literature \cite{FGPS-07}, and another one in which localization is prevented by topological considerations. This can be demonstrated by a $3\times3$ system, whose theory is developed in the current section. 

The three modes in the TR symmetry breaking setting (or in a TR symmetry setting assuming that both blocks do not communicate) physically represent a protected mode $0$ with wavenumber $\zeta_0=E>0$ and two modes $\pm1$ with wavenumber $\pm\zeta_1=\pm\sqrt{E^2-\eps_1}$ in our simple block model. 

For concreteness \footnote{The general solution operator is $P(y',y)$ propagating $I$ at $y'$ to the solution at $y$. Then, $P^{-1}(y,y')=P(y',y)$ and we find equations for $P$ by differentiating either in the $y$ or $y'$ variables. The closed-form equation for $(R_+,T_+)$ or $(R_-,T_-)$ is obtained by differentiating on the side where reflection occurs.}, we assume that the modes $0$ and $1$ propagate from top to bottom and the mode $-1$ from bottom to top (as $y$ increases). The propagator $P(y)$ propagates $I$ on the bottow (at $y=0$) to $P(y)$ on the top of a slab of thickness $y$. We now introduce some relationships between the scattering coefficients in the scattering matrix $S$ and the coefficients in $P$. 

We thus obtain the following relations between $S$ and $P$:
\begin{displaymath}
  P\left(\begin{matrix}
T_+ & R_- \\  0 & I
\end{matrix}\right) = \left(\begin{matrix}
I & 0 \\ R_+ & T_-
\end{matrix}\right).
\end{displaymath}
Solving for $P$ yields 
\begin{displaymath}
  P = \left(\begin{matrix}
  T_+^{-1} &  -T_+^{-1}R_-\\  R_+T_+^{-1} &T_--R_+T_+^{-1}R_-
\end{matrix}\right).
\end{displaymath}
From the unitarity of the scattering matrix, we may verify that $T_+-R_-T_-^{-1} R_+=T_+^{-*}$ and $T_--R_+T_+^{-1}R_-=T_-^{-*}$ although we will not use these relations.
Since our primary objective is the analysis of $(T_+,R_+)$, and more precisely $R_+$, which provides how the $0$ and $1$ modes are reflected back into the $-1$ mode, we concentrate on the left columns of $P$ and define
\begin{displaymath}
  P_{3\times2} = \left(\begin{matrix}
  T_+^{-1} \\  R_+T_+^{-1}
\end{matrix}\right).
\end{displaymath}
The above algebraic manipulations are independent of dimension. In the $3\times3$ problem, the matrix $T_+$ is a $2\times2$ matrix while $R_+$ is a $1\times2$ (co-)vector. The matrix $P_{3\times2}$ satisfies the same equation as $P$. 

We write $P_{3\times2}=(P_{ij})_{1\leq i\leq3;\ 1\leq j\leq2}$ and apply $\mL$ to functions that depend only on those variables and their complex conjugates. Since $P_{3\times2}$ satisfies a closed form equation, we are guaranteed that $\mL$ applied to such functions will depend only on the $P_{ij}$ of the $3\times2$ system. Significant simplifications are in fact possible. We know that closed Riccati equations  can be obtained for the reflection coefficients \cite{FGPS-07}. This justifies looking for a diffusion equation that only involves (functions of) the reflection coefficients. This is done as follows.

With the above convention, we may write the explicit form of the transmission and reflection matrices:
\begin{displaymath}
  T_+ = \left(\begin{matrix}
  t_1 & t_2 \\ t_3 & t_4
\end{matrix}\right) = \left(\begin{matrix}
  t_+^{1\to1} & t_+^{0\to1} \\ t_+^{1\to0} & t_+^{0\to0} 
\end{matrix}\right),\qquad R_+ = \left(\begin{matrix}
  r_1 & r_2 
\end{matrix}\right) = \left(\begin{matrix}
r_+^{1\to-1} & r_+^{0\to-1}
\end{matrix}\right).
\end{displaymath}
With the above expression for $P_{3\times2}$ we solve to obtain
\begin{displaymath}
  T_+ = \frac{1}{\det_{12}} \left(\begin{matrix}
  P_{22} & -P_{12} \\ -P_{21} & P_{11}
\end{matrix}\right),\quad R_+ = \left(\begin{matrix}
  \dfrac{\det_{32}}{\dt_{12}} & \dfrac{\dt_{13}}{\det_{12}}
\end{matrix}\right),\quad \dt_{pq} = P_{p1}P_{q2}-P_{q1}P_{q2}.
\end{displaymath}

Perhaps somewhat surprisingly, it turns out that we can obtain a closed form diffusion in the variables $(d_1,d_2)$ with $d_j=d_j(P,\bar P)=|r_j|^2=r_j\bar r_j$. This will help us solve the $3\times3$ problem as well as the $2\times2$ problem simply by assuming that the $0$ mode does not couple to the $2\times2$ system composed of the $\pm1$ modes. 

We first verify
\begin{displaymath}
  X_0 d_j = X_0 r_j \ \bar r_j + r_j X_0 \bar r_j =0.
\end{displaymath}

So, the objective is to apply $X_{kl}$ and $\bar X_{kl} X_{kl}$ to functions of $(d_1,d_2)$ and realize that the coefficients depend on the coefficients $P_{ij}$ only via the coefficients $(d_1,d_2)$.  The details of the derivation are as follows. Recall that $Y_{pq}=P_q\cdot\nabla_{P_p}=P_{q1}\partial_{P_{p1}}+P_{q2}\partial_{P_{p2}}$ with $P_q$ the $q$th row of $P_{3\times2}$. We calculate
\begin{displaymath}
  Y_{qp} \dt_{mn} = \dt_{qn}\delta_{pm}+\dt_{mq}\delta_{pn},
\end{displaymath}
from which we deduce
\begin{displaymath}
\begin{array}{rcl	}
Y_{pq}r_1 &=& \dfrac{\dt_{q2}\delta_{p3}+\dt_{3q}\delta_{p2}}{\dt_{12}} - \dfrac{\dt_{32}}{(\dt_{12})^2}(\dt_{q2}\delta_{p1}+\dt_{1q}\delta_{p2})\\
Y_{pq}r_2 &=& \dfrac{\dt_{q3}\delta_{p1}+\dt_{1q}\delta_{p3}}{\dt_{12}} - \dfrac{\dt_{13}}{(\dt_{12})^2}(\dt_{q2}\delta_{p1}+\dt_{1q}\delta_{p2}).
\end{array} 
\end{displaymath}
With this, we find
\begin{displaymath}
\begin{array}{rcl	}
  Y_{12}r_1=0,\ Y_{21}r_1 = -r_2, \ Y_{12}r_2=-r_1,\ Y_{21} r_2 =0\\
  Y_{13}r_1=-r_1^2,\ Y_{31}r_1 = 1, \ Y_{13}r_2=-r_1r_2,\ Y_{31} r_2 =0\\
  Y_{23}r_1=-r_1r_2,\ Y_{32}r_1 = 0, \ Y_{23}r_2=-r_2^2,\ Y_{32} r_2 =1.
\end{array}   
\end{displaymath}
Now
\begin{displaymath}
  X_{pq} \phi (d_1,d_2) = \partial_1 \phi X_{pq} d_1 + \partial_2 \phi X_{pq} d_2
\end{displaymath}
so that 
\begin{displaymath}
 \bar X_{pq} X_{pq}  = |X_{pq}d_1|^2 \partial_1^2 + (\overline{X_{pq} d_2} X_{pq}d_1+c.c.) \partial^2_{12} +|X_{pq}d_2|^2 \partial_2^2 + [\bar X_{pq}X_{pq} d_1] \partial_1 + [\bar X_{pq}X_{pq}d_2] \partial_2.
\end{displaymath}

Recall that $X_{pq}=\lambda_pY_{pq}-\lambda_q\bar Y_{qp}$ with $\lambda_1=\lambda_2=-\lambda_3=1$. With the above expressions for $Y_{pq}$, we obtain after some algebra
\begin{displaymath}
  |X_{12}d_1|^2=|X_{12}d_2|^2=d_1d_2, \ \overline{X_{12}d_2} X_{12}d_1 = -d_1d_2,\ \bar X_{12}X_{12}d_1=-\bar X_{12}X_{12}d_2=d_2-d_1.
\end{displaymath}
For the second vector field, we find
\begin{displaymath}
  |X_{23}d_1|^2=d_1^2d_2,\ |X_{23}d_2|^2 = d_2(d_2-1)^2,\ \overline{X_{23}d_2} X_{23}d_1 =d_1d_2(d_2-1)
\end{displaymath}
while the drift terms are given by
\begin{displaymath}
    \bar X_{23}X_{23} d_1 = d_1(d_2-1),\quad \bar X_{23} X_{23} d_2 = (d_2-1)^2.
\end{displaymath}
The third vector field $X_{13}$ is obtained by symmetry considerations with $d_1$ and $d_2$ exchanged compared to the above results for $X_{23}$. This is logical as reflection from  $0$ to $-1$ is physically similar to reflection from $1$ to $-1$. Note that since all the above quantities are real-valued, we find that $X_{kl}\bar X_{kl}\phi=\bar X_{kl} X_{kl}\phi$.

Let us define $\gamma_{kl}=|W_{kl}|^2$ and assume that $4\hat R_0=1$ or absorb that latter constant into $\gamma_{kl}$. With this, we just found that applied to functions $\phi(d_1,d_2)$, the infinitesimal generator in \eqref{eq:generator} reduced to
\begin{equation}\label{eq:gen3}
\begin{array}{rcl}
  \mL &=& \gamma_{12} \Big[ d_1d_2 \left(\begin{matrix}
  1&-1\\-1&1 
\end{matrix}\right): \nabla^2 + (d_2-d_1) \left(\begin{matrix}
 1\\-1
\end{matrix}\right)\cdot\nabla\Big]\\
    &+& \gamma_{23} \Big[ \left(\begin{matrix}
  d_1^2d_2 & d_1d_2(d_2-1) \\ d_1d_2(d_2-1) & d_2(d_2-1)^2
\end{matrix}\right): \nabla^2 + \left(\begin{matrix}
  d_1(d_2-1) \\ (d_2-1)^2
\end{matrix}\right)\cdot\nabla\Big]\\
    &+& \gamma_{13} \Big[ \left(\begin{matrix}
  d_1(d_1-1)^2 & d_1d_2(d_1-1) \\ d_1d_2(d_1-1) & d_2^2d_1
\end{matrix}\right): \nabla^2 + \left(\begin{matrix}
   (d_1-1)^2\\ d_2(d_1-1)
\end{matrix}\right)\cdot\nabla\Big].
\end{array}    
\end{equation}
We may recast the above expression as 
\begin{displaymath}
  \begin{array}{rcl}
 \mL= \gamma_{12}[d_1d_2 \varphi_3\otimes\varphi_3 : \nabla^2 + (d_2-d_1) \varphi_3\cdot\nabla]
 + \gamma_{13}[d_2\varphi_2\otimes\varphi_2 : \nabla^2 + (d_2-1) \varphi_2\cdot\nabla]
 \\ + \gamma_{23}[d_1 \varphi_1\otimes\varphi_1 : \nabla^2 + (d_1-1)\varphi_1\cdot\nabla]
 \\
  \varphi_3=\left(\begin{matrix}
1\\-1
\end{matrix}\right) ,\qquad
\varphi_2=\left(\begin{matrix}
     d_1 \\ d_2-1
\end{matrix}\right),\qquad
\varphi_1 =\left(\begin{matrix}
  d_1-1 \\ d_2
\end{matrix}\right).
\end{array}
\end{displaymath}
So, the generator may be expressed as the sum of three one-dimensional diffusions; $\gamma_{12}$ in the direction parallel to the edge $d_1+d_2=1$, $\gamma_{13}$ in the direction toward the point $(0,1)$, and $\gamma_{23}$ in the direction toward the point $(1,0)$. The combination provides a net drift in the direction $d_1+d_2=1$, which is never attained but converged to exponentially as we shall see below in a simplified setting. 

Let us first define the domain of definition for such an equation. We deduce from the unitarity of the scattering matrix that 
\begin{displaymath}
  |r_+^{1\to-1}|^2 + |r_+^{0\to-1}|^2+|t_-^{-1\to-1}|^2=1
\end{displaymath}
as a consequence of current conservation. As a result, $d_1+d_2+|t_-^{-1\to-1}|^2=1$ and the domain of definition of the diffusion is $0\leq d_1\leq 1$, $0\leq d_2\leq 1$, as well as $d_1+d_2\leq 1$. In other words, the two-dimension diffusion lives on an open triangle $T$ in $(d_1,d_2)$ space.

Each of the above three terms in \eqref{eq:gen3} is degenerate as a one-dimensional diffusion. However, the three terms combined ensure that $\mL$ is elliptic inside $T$ (with positive definite diffusion tensor at each point inside $T$). The equation nonetheless remains degenerate as the diffusion coefficients all vanish at the boundary of $T$. The rate at which they approach $0$ depends on the point at the boundary $\partial T$. In the vicinity of the two boundary segments defined by $d_1=0$ and $d_2=0$, we observe that the diffusion coefficients converge linearly to $0$ while the drift terms converge to a vector field pointing toward the inside of $T$. In the vicinity of the boundary segment $d_1+d_2=1$, the drift term vanishes linearly while the diffusion coefficients vanish quadratically. Heuristically, this implies that the diffusion may reach but is repelled from the former two segments while it converges asymptotically to the latter segment (as ``time" $L$ increases) while never reaching it. Neglecting transverse diffusion, the former segments are modeled locally by $x\partial^2_x+b\partial_x$ for $b>0$, where $x=0$ is the boundary point, while the latter segments are modeled locally by $x^2\partial^2_x+\alpha x\partial_x$ with $\alpha\in\Rm$ any constant. The standard change of variables $x=e^y$ shows that the point $x=0$ corresponds to the point $y=-\infty$ in a standard heat equation, and hence is never attained. 

Theories for equations of the form \eqref{eq:gen3} follow after some modifications of theories available in the literature. A complete theory for linearly vanishing coefficients in the vicinity of the boundary was developed in \cite{EM-PUP-13} and its references with applications primarily in biology. One dimensional versions of \eqref{eq:gen3} are treated in \cite[Chapter 7]{FGPS-07}. Here, although this is not absolutely necessary, we will consider the simplified setting where $\gamma=\gamma_{13}=\gamma_{23}$. Then, we observe that $\mL$ applied to functions $\phi(d_1+d_2)$ is of the form
\begin{displaymath}
  \mL \phi(d_1+d_2)  = \gamma \big[[(d_1+d_2)(1-d_1-d_2)^2 \partial^2 +(2-d_1-d_2)(1-d_1-d_2)\partial]\phi(d_1+d_2)\big].
\end{displaymath}
In other words, we are fortunate enough to obtain a one dimensional diffusion for the variable $\rho=d_1+d_2$ with
the diffusion operator, with $\gamma$ set to $1$ to simplify
\begin{displaymath}
  \mL_\rho = \rho(1-\rho)^2\partial^2_\rho + (1-\rho)(2-\rho)\partial_\rho.
\end{displaymath}
The above equation is similar to that in \cite[Chapter 7]{FGPS-07} (constructed for the variable $\tau=1-\rho$) with a simpler analysis. As in \cite[Chapter 7]{FGPS-07}, we observe that the law of the diffusion $\rho$ converges exponentially to $1$ (in the $L^1$ sense) as ``time" $L$ increases.

   More precisely, let $\pi_\rho$ be the solution of $\partial_L \pi_\rho = \mL_\rho^* \pi_\rho$ with initial condition $\pi_\rho(L=0,\rho)=\delta_{0^+}(\rho)$. Then $\pi_\rho$ converges to a delta function at $\rho=1$ exponentially rapidly as $L$ increases. Indeed, let $0\leq m_0(L)=\E(1-\rho)=\E|1-\rho|=\int_0^1(1-\rho)\pi_\rho(\rho)d\rho$. We find
\begin{displaymath}
    \partial_L m_0 = \int_0^1 \pi_\rho \mL(1-\rho)  d\rho= \int_0^1 (1-\rho)(\rho-2)\pi_\rho d\rho \leq -m_0(z)
\end{displaymath}
since $2-\rho\geq1$. This shows that $m_0=\E|1-\rho|$ is either $0$ or decays exponentially faster that $e^{-L}$. This shows that $\pi_\rho(L,\rho)$ converges in distribution exponentially to the atom $\delta_1(\rho)$ in any reasonable metric.

Armed with this information,  we return to the analysis of the diffusion in \eqref{eq:gen3} with $\gamma_{13}=\gamma_{23}$ to observe that $d_1+d_2$ converges to $1$ exponentially quickly. For large ``times" $L$, we may therefore approximate the law $\pi(L,d_1,d_2)$ solution of $\partial_L \pi = \mL^*\pi$ as $\pi(L,d_1,d_2)\approx \pi_{d_1}(L,d_1)\frac1{\sqrt2}\delta(d_1+d_2-1)$, where we find that the reduced  density $\pi_{d_1}$ is normalized to $1$ and solves the reduced equation
\begin{displaymath}
  \partial_L \pi_{d_1} = \mL^*_{d_1} \pi_{d_1}, \qquad 
   \mL_{d_1} = (\gamma_{12}+\gamma) \partial_{d_1} \big[d_1(1-d_1) \partial_{d_1} \big].
\end{displaymath}
This is a degenerate diffusion equation in the variable $d_1$ posed on the interval $(0,1)$ of the type analyzed in \cite{EM-PUP-13}. Note that the diffusion coefficient converges linearly to $0$ as $d_1$ approaches either $0$ or $1$. We easily observe that there is an invariant measure (unique thanks to the results of \cite{EM-PUP-13}) solution of $\mL_{d_1}^*\pi_\infty=0$ and given by $\pi_\infty=1$ on $(0,1)$.

We summarize the above derivation as 
\begin{theorem}
\label{thm:gen3}
Let $\pi(L,d_1,d_2)$ be the solution of $\partial_L \pi= \mL^*\pi$ with initial conditions $\pi(0)=\delta_{0+}(d_1)\delta_{0+}(d_2)$. We assume that $\gamma_{13}=\gamma_{23}>0$.\footnote{The result of the theorem generalizes to the setting where $\gamma_{13}>0$ and $\gamma_{23}>0$ are not necessarily equal. These coefficients are in fact not expected to be equal since the $0$ and $\pm1$ modes have different currents $j_k$. We briefly sketch the proof of the generalization. We verify that the variable $m_0(L)=\E(1-d_1-d_2)=\E|1-d_2-d_2|$ also decays exponentially following the same proof as for the variable $\rho$ above. The rate of convergence depends on the coefficients $\gamma_{ij}$. This allows us to approximate $\pi(L,d_1,d_2)$ by $\pi_{d_1}(L,d_1)\frac 1{\sqrt2}\delta(d_1+d_2-1)$ as above. It remains to analyze the equation for $\pi_{d_1}(L,d_1)$, which is a degenerate diffusion equation treated in \cite{EM-PUP-13} and show that $\pi_{d_1}(L,d_1)$ converges to an invariant measure, which is no longer uniform on $(0,1)$ when $\gamma_{13}\not=\gamma_{23}$.}
Then as $L\to\infty$, the above probability distribution converges in distribution to the invariant measure $\frac1{\sqrt2}\delta(d_1+d_2-1)$.
\end{theorem}
{\bf Proof.} The Green's function $\pi(L,d_1,d_2)$ is defined and smooth for positive $L$ as in regular potential theory and as in the potential theory developed in \cite{EM-PUP-13} for degenerate coefficients at the domain's boundary. The details of adaptation of these theories  to the equation with a different degeneracy at the interface $d_1+d_2=1$ are not done in detail here. Then, $d_1+d_2$ satisfies an independent one dimensional diffusion equation with $d_1+d_2$ converging to $1$ as $L\to\infty$ as we demonstrated above. This shows that $\pi(L,d_1,d_2)$ also concentrates to the vicinity of $d_1+d_2=1$. In the $d_1$ variable, we obtained in the above derivation that the unique invariant measure was given by a constant distribution, which yields the theorem. \endproof

\medskip

The above result shows that for $L$ large (or equivalently for $\gamma$ large by rescaling), the reflection coefficients  $|r_+^{1\to-1}|^2$ and $|r_+^{0\to-1}|^2$ diffuse between the values of $0$ and $1$ and by definition of the invariant measure, come close to any value between $0$ and $1$ as $L$ progresses. This concludes our proof of Theorems \ref{thm:diff} and \ref{thm:diffTRS} in the topologically non-trivial cases.

\subsection{Case of $2\times2$ systems}
\label{sec:2by2}

We can use the calculations of the previous section to revisit standard results of localization in the presence of two modes $\pm1$. It suffices to neglect the mode $0$ in the preceding calculations and assume that it does not couple with the other modes. Only the quantity $d_1=|r_+^{1\to-1}|^2$ then matters and we know from scattering that
\begin{displaymath}
  d_1 + |t_-^{-1\to-1}|^2 = d_1 + |t_+^{1\to1}|^2 =1.
\end{displaymath}
So, proving that $d_1$ converges to $1$ exponentially as the thickness of the slab increases shows that modes localize: transmission is exponentially small as a function of thickness $L$. In the above equation, this corresponds to looking at $\gamma_{13}$ the only non vanishing coefficient and considering function $\phi=\phi(d_1)$.  We then find the diffusion
\begin{displaymath}
  \mL = \gamma_{13} [d_1(d_1-1)^2 \partial^2 + (1-d_1)^2\partial].
\end{displaymath}
Note that the drift term $(1-d_1)^2$ always pushes the diffusion towards $d_1=1$. However, since both diffusion and drift decay to $0$ quadratically at $d_1=1$, the limit is never attained.
Let $\tau=1-d_1$ the transmission coefficient. In this variable, we find
\begin{displaymath}
  \mL_\tau = \gamma_{13} [\tau^2(1-\tau)\partial^2_\tau - \tau^2 \partial_\tau].
\end{displaymath}
This is the operator describing transmission in \cite[Chapter 7]{FGPS-07}. We know that $\tau$ decays exponentially to $0$, or more precisely that $\pi(\tau)$ its law concentrates exponentially rapidly to the vicinity of $\tau=0$. This is a signature of the localization of waves in random slabs.

This concludes our proof of Theorems \ref{thm:diff} and \ref{thm:diffTRS} in the topologically trivial cases. 


\section{Conclusions}

This paper introduces a class of Hamiltonians in \eqref{eq:edgeH} modeling the low frequency components (for energies close to the Fermi energy) of general edge states at the interface of two-dimensional materials in different topological phases. After appropriate regularization, these Hamiltonians are classified in Theorem \ref{thm:topo} as Fredholm operators based on their index $\ind H_v=\Mtop-\Ntop$ given by the difference of zero modes propagating with positive and negative velocities along the edge, respectively.

In the presence of fermionic time reversal symmetry (TRS), $\Mtop=\Ntop$ and the above index is trivial. Another index given by $\ind_2 H_v=\Mtop$ mod $2$ separates edge Hamiltonians into two classes as described in Theorem \ref{thm:topoTRS}.

\medskip

The Hamiltonians are defined on a open domain $\Rm^2$ and are not required to satisfy any translational invariance. They are therefore amenable to perturbations by a large class of random fluctuations. The spectral decomposition of the unperturbed Hamiltonian $H_0$ in section \ref{sec:scat} allows one to develop a scattering theory for the propagating modes of $H$. In this paper, the random fluctuations are modeled by a specific operator $V$ so that the evanescent modes do not couple with the propagating ones. 

Under this assumption, we were able to asses the influence of the topology of $H_v$ on the scattering matrix. More specifically, we show that the transmission (conductance) ${\rm Tr}T_+^*T_+$ is bounded from below by $\Mtop-\Ntop$, and, for a certain choice of mode couplings, is asymptotically equal to that value in the limit of strong random fluctuations. Transmission, and hence the absence of Anderson localization, is one of the hallmarks of non trivial edge Hamiltonians. We also obtain that backscattering is present for energies $E^2$ above a certain threshold (equal to $\eps_1$ in our model). Only for specific, random, linear combinations of the propagating modes do we observe a total absence of backscattering; see Theorem \ref{thm:scatTRB}.

For TRS Hamiltonians, the same scattering picture emerges with $\ind H_v=\Mtop-\Ntop$ replaced by the $\Zm_2$ index $\ind_2 H_v=\Mtop$ mod $2$. We obtain  ${\rm Tr}T_+^*T_+$ and ${\rm Tr}T_-^*T_-$ are bounded from below by $\Mtop$ mod $2$ and when the latter is non-trivial, obtain the existence of random linear combinations of propagating modes (one for each direction of propagation) such that no-backscattering occurs. 

\medskip

The model for the scattering amplitudes takes the form of a system of one dimensional ordinary differential equations obeying a current conservation. The macroscopic limit for the influence of highly oscillatory random fluctuations is then well described by a diffusion equation. We generalized known derivations of diffusion equations to the specific $3\times3$ systems that naturally appear when zero modes are coupled with a pair of other propagating modes. This allowed us to obtain the following result (for a specific choice of random fluctuations): In the high scattering regime, only $\ind H_v$ (in the general case) or $\ind_2 H_v$ (in the TRS case) modes propagate without any back-scattering. All other modes are localized as in standard Anderson localization in the sense that their transmission decays exponentially with the thickness of the random slab (or equivalently with the strength of the random fluctuations).

\section*{Acknowledgments} This work was partially supported by the National Science Foundation and the Office of Naval Research. The author would like to acknowledge multiple stimulating discussions with a number of researchers including Gian Michele Graf, Harish Krishnaswamy, Mikael Rechtsman, Hermann Schulz-Baldes, and Michael Weinstein.

%
%
%
%
%
%
%
%
%
%
%


\end{document}